\documentclass[11pt]{article}

\usepackage{amsfonts,latexsym,amsthm,amssymb,amsmath,amscd,euscript,tikz,mathtools}
\usepackage{framed}
\usepackage{fullpage}
\usepackage{color}
\usepackage[colorlinks=true,citecolor=blue,linkcolor=blue]{hyperref}
\usepackage{enumitem}
\usepackage{siunitx}
\usepackage{textcomp}
\usepackage{physics}
\usepackage{mathrsfs}
\usepackage{dsfont}
\usepackage[ruled,vlined]{algorithm2e}
\usepackage{titlesec}
\usepackage{tikz-cd}
\usepackage{thm-restate}
\usepackage{thmtools}
\usepackage{cleveref}
\usepackage{graphicx}
\usetikzlibrary{positioning,chains,fit,shapes,calc}

\allowdisplaybreaks[1]

\newtheorem{theorem}{Theorem}
\newtheorem{proposition}[theorem]{Proposition}
\newtheorem{lemma}[theorem]{Lemma}

\newtheorem{corollary}[theorem]{Corollary}

\theoremstyle{definition}
\newtheorem{definition}[theorem]{Definition}

\newtheorem{remark}[theorem]{Remark}

\theoremstyle{remark}

\newcommand{\cA}{\mathcal{A}}\newcommand{\cB}{\mathcal{B}}
\newcommand{\cC}{\mathcal{C}}

\newcommand{\cH}{\mathcal{H}}
\newcommand{\cI}{\mathcal{I}}

\newcommand{\cY}{\mathcal{Y}}

\newcommand{\bC}{\mathbb{C}}
\newcommand{\bF}{\mathbb{F}}

\newcommand{\bN}{\mathbb{N}}

\newcommand{\bZ}{\mathbb{Z}}

\newcommand{\1}{\mathds{1}}

\newcommand{\poly}{\operatorname{poly}}

\newcommand{\Cay}{\operatorname{Cay}}

\newcommand{\Dec}{\operatorname{Dec}}

\newcommand{\CSS}{\operatorname{CSS}}
\newcommand{\dis}{\operatorname{dis}}

\newcommand{\Ham}{\mathbf{H}}
\newcommand{\clustX}[1]{Y_X^{#1}}
\newcommand{\clustZ}[1]{Y_Z^{#1}}
\newcommand{\transX}[1]{\cY_X^{#1}}
\newcommand{\transZ}[1]{\cY_Z^{#1}}
\newcommand{\allones}{\mathbf{1}}
\newcommand{\Tan}{\operatorname{Tan}}

\newcommand{\nc}{\newcommand}

\nc{\on}{\operatorname}
\nc{\Spec}{\on{Spec}}
\nc{\Aut}{\textit{Aut}}
\nc{\id}{\textit{id}}
\nc{\chr}{\on{char}}
\nc{\im}{\on{im}}
\nc{\Hom}{\on{Hom}}
\nc{\lcm}{\on{lcm}}
\nc{\dual}[1]{\prescript{t}{}{#1}}
\nc{\transpose}[1]{{#1}^{\intercal}}
\nc{\Sym}{\on{Sym}}
\nc{\End}{\on{End}}
\nc{\stab}{\on{stab}}
\nc{\Li}{\on{Li}}
\nc{\spn}{\on{span}}
\nc{\sgn}{\on{sgn}}
\nc{\supp}{\on{supp}}
\nc{\Unif}{\on{Unif}}

\title{NLTS Hamiltonians and Strongly-Explicit SoS Lower Bounds \\  from Low-Rate Quantum LDPC Codes}
\author{Louis Golowich \\
  UC Berkeley \\
  \href{mailto:lgolowich@berkeley.edu}{\texttt{lgolowich@berkeley.edu}}
  \and
  Tali Kaufman \\
  Bar-Ilan University \\
  \href{mailto:kaufmant@mit.edu}{\texttt{kaufmant@mit.edu}}
}

\begin{document}

\pagenumbering{gobble}

\maketitle


\begin{abstract}
Recent constructions of the first asymptotically good quantum LDPC (qLDPC) codes led to two breakthroughs in complexity theory: the NLTS (No Low-Energy Trivial States) theorem (Anshu, Breuckmann, and Nirkhe, STOC'23), and explicit lower bounds against a linear number of levels of the Sum-of-Squares (SoS) hierarchy (Hopkins and Lin, FOCS'22).

In this work, we obtain improvements to both of these results using qLDPC codes of \textit{low rate}:
\begin{itemize}
\item Whereas Anshu et al.~only obtained NLTS Hamiltonians from qLDPC codes of linear dimension, we show the stronger result that qLDPC codes of arbitrarily small positive dimension yield NLTS Hamiltonians. 
\item The SoS lower bounds of Hopkins and Lin are only weakly explicit because they require running Gaussian elimination to find a nontrivial codeword, which takes polynomial time. We resolve this shortcoming by introducing a new method of planting a strongly explicit nontrivial codeword in linear-distance qLDPC codes, which in turn yields strongly explicit SoS lower bounds.
\end{itemize}
Our ``planted'' qLDPC codes may be of independent interest, as they provide a new way of ensuring a qLDPC code has positive dimension without resorting to parity check counting, and therefore provide more flexibility in the code construction.

\end{abstract}




\pagenumbering{arabic}

\section{Introduction}
Recent breakthrough constructions of asymptotically good quantum LDPC (qLDPC) codes \cite{panteleev_asymptotically_2022,leverrier_quantum_2022-1,dinur_good_2023} have led to major advances in complexity theory. Specifically, Anshu et al.~\cite{anshu_nlts_2023} applied these codes to prove the NLTS theorem, which provides perhaps the most significant progress to date towards the quantum PCP conjecture. Meanwhile, Hopkins and Lin~\cite{hopkins_explicit_2022-1} applied the same codes to obtain the first explicit lower bounds against a linear number of levels of the Sum-of-Squares semidefinite programming (SoS SDP) hierarchy, which is one of the most powerful algorithmic frameworks for approximating the satisfiability of constraint satisfaction problems (CSPs).

In this paper, we improve upon both of these complexity theoretic results. Along the way, we introduce a new method for ensuring a qLDPC code has positive dimension, which may be of independent interest. Our contributions are therefore threefold:
\begin{enumerate}
\item \textbf{NLTS Hamiltonians from low-rate codes:} The breakthrough construction of NLTS Hamiltonians of \cite{anshu_nlts_2023} from asymptotically good qLDPC codes relied on both the linear dimension and distance of the codes. A promising approach \cite{nirkhe_making_2023} for further progress towards qPCP is to construct more general NLTS Hamiltonians with additional properties. We make progress in this direction by constructing NLTS Hamiltonians from qLDPC codes of arbitrary positive dimension, thereby removing the linear-dimension requirement in \cite{anshu_nlts_2023}. Our result highlights the usefulness of local Hamiltonians with low-dimensional ground spaces for studying qPCP. Our proof leverages techniques of \cite{eldar_local_2017}, which conjecturally constructed NLTS Hamiltonians from linear-distance quantum locally testable codes of arbitrary positive dimension (which are not known to exist). However, we obtain the NLTS property without assuming local testability nor linear dimension. Instead, the key ingredient ensuring NLTS Hamiltonians is a small-set expansion property of the qLDPC codes.
\item \textbf{Planted quantum LDPC codes:} We show how to plant an explicit nontrivial codeword in a linear-distance qLDPC code, which may have otherwise had rate $0$. To the best of our knowledge, this construction yields the first linear-distance qLDPC codes for which nontrivial dimension is established without resorting to parity-check counting. It has been an open question in the literature to develop new such techniques for bounding dimension (see for instance Section~1.1 of \cite{dinur_new_2023}, and also \cite{dikstein_locally_2020}).
\item \textbf{Strongly explicit SoS lower bounds:} We apply our planted qLDPC codes to obtain the first \textit{strongly} explicit family of CSPs that cannot be refuted by a linear number of levels of the SoS hierarchy. This result strengthens the work of \cite{hopkins_explicit_2022-1}, which provided the first \textit{weakly} explicit construction of such an SoS lower bound using qLDPC codes. Our improvement stems from the fact that our planted codes have planted codeword given by the all-1s vector, which is strongly explicit.
\end{enumerate}

These results together show new ways to both construct and apply qLDPC codes of low rate. In the remainder of this section, after providing some background on qLDPC codes, we describe each of these results in more depth. We then discuss open questions that arise from our results.

\subsection{Background on qLDPC Codes}
This section provides some definitions we will need to state our results. The quantum codes we consider in this paper are quantum CSS codes. An $n$-qudit CSS code $\cC=\CSS(C_X,C_Z)$ of alphabet size (i.e.~local dimension) $q$ is defined by a pair of classical codes $C_X,C_Z\subseteq\bF_q^n$ such that $C_X^\perp\subseteq C_Z$. The associated quantum code is then given by $\cC=\spn\{\sum_{y'\in C_X^\perp}\ket{y+y'}:y\in C_Z\}$. This code has dimension $k=\dim(C_Z)-\dim(C_X^\perp)$ and distance $d=\min_{y\in(C_Z\setminus C_X^\perp)\cup(C_X\setminus C_Z^\perp)}|y|$, meaning it encodes a $k$-qudit message into an $n$-qudit code state, and the message can be recovered from any $n-(d-1)$ code qudits. We assume $C_X=\ker H_X,C_Z=\ker H_Z$ for associated parity check matrices $H_X\in\bF_q^{m_X\times n},H_Z\in\bF_q^{m_Z\times n}$. If every row and colum of $H_X$ and $H_Z$ has Hamming weight $\leq\ell$, we say that $\cC$ has locality $\ell$. A family of qLDPC codes is a family of codes with constant locality $\ell$ and growing block length $n$.

It was a longstanding open question to construct linear-distance qLDPC codes. This question was resolved by Panteleev and Kalachev \cite{panteleev_asymptotically_2022}, who obtained qLDPC codes of linear distance and linear dimension. Subsequent works \cite{leverrier_quantum_2022-1,dinur_good_2023} provided additional related constructions.

These codes in fact possess\footnote{\cite{hopkins_explicit_2022-1} were the first to consider small-set (co)boundary expansion for linear-distance qLDPC codes, and showed that the codes of \cite{leverrier_quantum_2022-1} possess this property. \cite{dinur_good_2023} later constructed additional good qLDPC codes for which they proved this expansion property. We explain at the end of Section~\ref{sec:ssexp} why the decoder of \cite{leverrier_efficient_2023,leverrier_decoding_2023} implies that the codes of \cite{panteleev_asymptotically_2022} also possess this expansion property.} the following stronger notion of distance, which guarantees that all low-weight errors have syndromes whose weight is linear in the error weight (as opposed to just having nonzero syndromes). Below, for a code $C$, we denote $|y|_C=\min_{y'\in C}|y+y'|$.

\begin{definition}[Small-set (co)boundary expansion; restatement of Definition~\ref{def:ssexp}]
  \label{def:ssexpinf}
  Let $\cC=\CSS(C_X=\ker H_X,C_Z=\ker H_Z)$ be a CSS code given by parity check matrices $H_X\in\bF_q^{m_X\times n}$ and $H_Z\in\bF_q^{m_Z\times n}$. For $c_1,c_2>0$, we say that $\cC$ has \textbf{$(c_1,c_2)$-small-set boundary expansion} if it holds for every $y\in\bF_q^n$ with $|y|\leq c_1n$ that
  \begin{equation*}
    \frac{|H_Zy|}{m_Z} \geq c_2\frac{|y|_{C_X^\perp}}{n}.
  \end{equation*}
  Similarly, $\cC$ has \textbf{$(c_1,c_2)$-small-set coboundary expansion} if it holds for every $y\in\bF_q^n$ with $|y|\leq c_1n$ that
  \begin{equation*}
    \frac{|H_Xy|}{m_X} \geq c_2\frac{|y|_{C_Z^\perp}}{n}.
  \end{equation*}
\end{definition}

This notion of small-set (co)boundary expansion underlies both the NLTS Hamiltonians of \cite{anshu_nlts_2023} and the SoS lower bounds of \cite{hopkins_explicit_2022-1}. Note that a code with $(c_1,c_2)$-small set boundary and coboundary expansion by definition has distance $\geq c_1n$.

\subsection{NLTS Hamiltonians from Low-Rate qLDPC Codes}
The quantum PCP (qPCP) conjecture, which states that it is QMA-hard to compute a constant-factor approximation to the ground energy of a local Hamiltonian, is a major open question in quantum complexity theory that has remained largely elusive. Perhaps the most significant progress towards this conjecture was the NLTS theorem, which was recently proven by Anshu, Breuckmann, and Nirkhe \cite{anshu_nlts_2023} using an application of asymptotically good qLDPC codes. This result provides a family of local Hamiltonians that have ``no low-energy trivial states'' (NLTS), where a trivial state is one computed by a constant-depth circuit. The NLTS theorem therefore provides local Hamiltonians exhibiting a weaker form of hardness of approximation than required by qPCP, and is indeed a necessary consequence of the qPCP conjecture under the widely believed assumption that $\text{NP}\neq\text{QMA}$.

Anshu et al.~\cite{anshu_nlts_2023} constructed their NLTS Hamiltonians using the asymptotically good quantum Tanner codes of \cite{leverrier_quantum_2022-1}. In particular, their proof of NLTS relied on the codes having both linear distance and dimension. It was an open question whether such linear dimension was necessary for NLTS. This question is motivated by the suggestion \cite{nirkhe_making_2023} that constructing more general families of NLTS Hamiltonians may lead to further progress towards the qPCP conjecture. Furthermore, some earlier partial progress towards NLTS used codes of smaller dimension \cite{eldar_local_2017}, which again raises the question of whether linear dimension is necessary. Our main result on NLTS resolves this question, as we obtain NLTS Hamiltonians from qLDPC codes of arbitrary positive dimension.

NLTS Hamiltonians are formally defined as follows. Recall that a family of Hamiltonians is \textit{$\ell$-local} if every $\Ham$ in the family can be expressed as a sum of Hamiltonians, each of which act nontrivially on $\leq\ell$ qubits. If $\ell=O(1)$ we say the family is \textit{local}. We also say that a state $\rho$ is an \textit{$\epsilon$-approximate ground state} of a Hamiltonian $\Ham\succeq 0$ if $\Tr(\rho\Ham)\leq\epsilon$.

\begin{definition}[NLTS Hamiltonians]
  A family of local Hamiltonians $(\Ham_n)_{n\rightarrow\infty}$ with $0\preceq\Ham_n\preceq I$ is \textbf{NLTS} if there exists $\epsilon>0$ such that the minimum depth of any quantum circuit computing an $\epsilon$-approximate ground state of $\Ham_n$ approaches $\infty$ as $n\rightarrow\infty$.
\end{definition}

Recall that for a CSS code $\cC=\CSS(C_X=\ker H_X,C_Z=\ker H_Z)$, the associated code Hamiltonian is given by
\begin{equation*}
  \Ham = \frac12(\Ham_X+\Ham_Z)
\end{equation*}
for
\begin{align*}
  \Ham_X &= \frac{1}{m_X}\sum_{y\in\mathrm{rows}(H_X)}\frac{I-X^y}{2} \\
  \Ham_Z &= \frac{1}{m_Z}\sum_{y\in\mathrm{rows}(H_Z)}\frac{I-Z^y}{2},
\end{align*}
where $X$ and $Z$ denote the respective Pauli operators. Thus in particular the ground space of $\Ham$ is precisely the code space $\cC=\spn\{\sum_{y'\in C_X^\perp}\ket{y+y'}:y\in C_Z\}$.

Anshu et al.~\cite{anshu_nlts_2023} showed that for every family of qLDPC codes with linear dimension and constant small-set boundary and coboundary expansion, the associated code Hamiltonians are NLTS. Thus for instance the quantum Tanner codes of \cite{leverrier_quantum_2022-1} yield NLTS code Hamiltonians.

Our result below improves upon this result of \cite{anshu_nlts_2023} by removing the linear dimension requirement.

\begin{theorem}[NLTS from low-rate codes; informal statement of Corollary~\ref{cor:nlts}]
  \label{thm:nltsinf}
  Let $(\cC^{(n)})_{n\rightarrow\infty}$ be an infinite family of qLDPC codes over the alphabet\footnote{For simplicity we restrict attention to the binary alphabet $\bF_2$ in our proof of Theorem~\ref{thm:nltsinf}, though we suspect the result should extend to arbitrary alphabets $\bF_q$.} $\bF_2$ of block length $n$ and positive dimension that have $(c_1,c_2)$-small set boundary and coboundary expansion for some constants $c_1,c_2>0$. Then the family of associated code Hamiltonians $(\Ham^{(n)})_{n\rightarrow\infty}$ is NLTS.
\end{theorem}

Our proof of Theorem~\ref{thm:nltsinf} follows the general framework of \cite{eldar_local_2017,anshu_nlts_2023} in showing circuit lower bounds for code Hamiltonians. Specifically, Eldar and Harrow \cite{eldar_local_2017} showed that in order to show the code Hamiltonians $\Ham$ are NLTS, it suffices to show that every distribution obtained by measuring an approximate ground state of $\Ham$ in either the $X$ or $Z$ basis is \textit{well spread}. Here a distribution $D$ over $\bF_2^n$ is well spread if there exist sets $S_0,S_1\subseteq\bF_2^n$ separated by a linear Hamming distance $\dis(S_0,S_1)\geq\Omega(n)$ such that $D$ assigns constant probability $D(S_0),D(S_1)\geq\Omega(1)$ to both sets.

Both \cite{eldar_local_2017,anshu_nlts_2023} show this well-spreadness property for code Hamiltonians by combining a distance/expansion property of the code with an uncertainty principle. However, the two works different use assumptions on the code as well as different uncertainty principles:
\begin{itemize}
\item \cite{eldar_local_2017} assumes the code is locally testable and of linear distance, which implies the approximate ground states have a certain linear structure. They then use an uncertainty principle (see Lemma~\ref{lem:uncertainty}) that is able to leverage this linear structure and prove well-spreadness regardless of the code dimension.
\item \cite{anshu_nlts_2023} assumes the code has small-set boundary and coboundary expansion, which is weaker than local testability and therefore yields less structure in the approximate ground states. They then use a different uncertainty principle with which they are still able to prove well-spreadness, but only for codes of linear dimension.
\end{itemize}
Because linear-distance quantum locally testable codes are not known to exist, the NLTS Hamiltonians of \cite{eldar_local_2017} remain conjectural.

We prove Theorem~\ref{thm:nltsinf} by combining these two approaches: we make the weaker assumption that our code has small-set boundary and coboundary expansion, but show that the approximate ground states still have enough linear structure to apply the uncertainty principle in Lemma~\ref{lem:uncertainty}. We then conclude that the code Hamiltonians are NLTS regardless of the code dimension.

At the core of our argument is the application of a ``decoding'' procedure for approximate ground states of codes with small-set (co)boundary expansion, which is unintuitive in the sense that far-apart approximate ground states may decode to the the same true ground state. However, we are able to show that in some sense, the low-energy space of the code Hamiltonian acts similarly enough to a true code space that the argument still goes through.

\subsection{Planted Quantum LDPC Codes}
\label{sec:plantinf}
This section presents our result on planting a nontrivial codeword in qLDPC codes.

The recent breakthrough constructions of linear-distance qLDPC codes (\cite{panteleev_asymptotically_2022}, followed by \cite{leverrier_quantum_2022-1,dinur_good_2023}) all bound the code dimension by counting parity checks. Specifically, these works use the fact that if $\cC=\CSS(C_X=\ker H_X,C_Z=\ker H_Z)$ for $H_X\in\bF_q^{m_X\times n},H_Z\in\bF_q^{m_Z\times n}$, then $\cC$ has dimension $k\geq n-m_X-m_Z$. However, this bound may not be tight if there are redundant parity checks in $H_X,H_Z$. Indeed, it has been an open question in the coding theory literature to provide new ways of ensuring that LDPC codes have positive dimension; for instance, this question was of central importance in the code constructions of \cite{dikstein_locally_2020,dinur_new_2023}.

Our result below makes progress on this question, by showing how to plant a nontrivial codeword in the linear-distance quantum Tanner codes of \cite{leverrier_quantum_2022-1}. In fact, we show that like the codes of \cite{leverrier_quantum_2022-1} our planted codes possess small-set (co)boundary expansion.

\begin{theorem}[Planted quantum Tanner codes; restatement of Theorem~\ref{thm:plantedresult}]
  \label{thm:plantedresultinf}
  For every finite field $\bF_q$, there exist constants $c_1,c_2>0$ such that there is a strongly explicit infinite family $(\cC^{(n)})_{n\rightarrow\infty}$ of quantum LDPC CSS codes for which every $\cC^{(n)}=\CSS(C_X^{(n)},C_Z^{(n)})$ with $C_X^{(n)},C_Z^{(n)}\subseteq\bF_q^n$ has the following properties:
  \begin{enumerate}
  \item $\cC^{(n)}$ has $(c_1,c_2)$-small-set boundary and coboundary expansion, and therefore has distance $\geq c_1n$.
  \item The all-1s vector $\allones\in\bF_q^n$ lies in $C_X^{(n)}\setminus {C_Z^{(n)}}^\perp$ and in $C_Z^{(n)}\setminus {C_X^{(n)}}^\perp$.
  \end{enumerate}
\end{theorem}

Theorem~\ref{thm:nltsinf} implies that the code Hamiltonians of our planted quantum Tanner codes over the binary alphabet $\bF_2$ in Theorem~\ref{thm:plantedresultinf} are NLTS. In Section~\ref{sec:sosinf} below, we present another complexity-theoretic application of these codes, namely to SoS lower bounds, which crucially relies on the their planted nature.


Our construction of planted quantum Tanner codes is motivated by a more basic classical analogue. Recall that a classical Tanner code is specified by a $\Delta$-regular graph $\Gamma$ and an inner code $C_{\text{in}}\subseteq\bF_q^\Delta$, where the code components correspond to edges of the graph, and the parity checks impose the constraint that the local view of each vertex is a codeword in $C_{\text{in}}$.

The standard method for ensuring a classical Tanner code $C$ has positive rate is to require $C_{\text{in}}$ to have sufficiently large rate $>1/2$, and then to bound the number of resulting linear constraints on $C$ from the parity checks. However, we may alternatively simply require that $C_{\text{in}}$ contain the all-1s vector $\allones\in\bF_q^\Delta$, so that $C$ then must contain the global all-1s vector $\allones\in\bF_q^n$. If $C$ contains no other nontrivial codewords, then it is a repetition code, which is typically unintersting classically.

However, we construct a quantum analogue of this construction, which is more nuanced, and has interesting complexity theoretic implications. Indeed, whereas classically it is easy to achieve linear distance and positive dimension by taking a repetition code, to the best of our knowledge the only known quantum LDPC codes of linear distance and positive dimension are the recent constructions of \cite{panteleev_asymptotically_2022,leverrier_quantum_2022-1,dinur_good_2023}, which can in fact achieve linear dimension.

Recall that a quantum Tanner code $\cC=\CSS(C_X,C_Z)$ \cite{leverrier_quantum_2022-1} is constructed by imposing constraints from a \textit{pair} of classical codes $C_A,C_B\subseteq\bF_q^\Delta$ on a \textit{square Cayley complex} $(V,E,Q)$, which is a graph $(V,E)$ with the additional high-dimensional structure of faces, or squares, in $Q$; the qudits of the code correspond to the $n=|Q|$ faces in $Q$.

To prove Theorem~\ref{thm:plantedresultinf}, we show that if we require the local all-1s vector $\allones\in\bF_q^\Delta$ to lie in $C_A$ and in $C_B^\perp$, and $q$ is relatively prime with $n$, then the global all-1s vector $\allones\in\bF_q^n$ lies in $C_Z\setminus C_X^\perp$ and $C_X\setminus C_Z^\perp$, so in particular $\cC=\CSS(C_X,C_Z)$ has dimension $\geq 1$.

The proof that $\allones\in C_A,C_B^\perp$ implies $\allones\in C_X,C_Z$ is immediate, as in the classical case. However, we prove that $\allones\notin C_X^\perp,C_Z^\perp$ using a parity (or more generally, arity) mismatch: we argue that $C_X^\perp$ and $C_Z^\perp$ are spanned by vectors whose components sum to $0\in\bF_q$, whereas the components of $\allones\in\bF_q^n$ do not sum to $0$ by the assumption that $q,n$ are relatively prime, so that the characteristic $p$ of $\bF_q$ does not divide $n$.

The requirement that $q,n$ are relatively prime requires some care to enforce. As $n=|Q|$ equals the number of faces in an expanding square Cayley complex, it must be a multiple of the order of a group on which there exist constant-degree Cayley expanders (see Section~\ref{sec:qtan}, and in particular Section~\ref{sec:dis}, for background on Cayley graphs and expansion). Therefore if we for instance focus on the $q=2$ case, we need families of Cayley expanders over groups of odd order. However, many well-known Cayley expanders, such as the Ramanujan graphs of \cite{lubotzky_ramanujan_1988} and \cite{morgenstern_existence_1994}, exclusively use groups of even order. We therefore instead use the Cayley expanders given in Example~3.4 of \cite{lubotzky_groups_1993}, for which the number of vertices is a power of any desired prime. While these graphs have constant degree and constant expansion, we amplify the expansion to be almost-Ramanujan using the techniques of \cite{jeronimo_almost_2022}.\footnote{This expansion amplification may be stronger than necessary, but for consistency with prior works and simplicity of presentation, it is convenient for us to have almost-Ramanujan expansion.}

We still must show that the resulting planted quantum Tanner codes have good small-set (co)boundary expansion and therefore good distance. By the results of \cite{leverrier_quantum_2022-1}, it suffices to show that the inner codes $(C_A,C_B)$ can be chosen to possess a property called \textit{product-expansion} (Definition~\ref{def:prodexp}). This property was shown for random inner codes by \cite{kalachev_two-sided_2023,dinur_good_2023}; we extend the proof of \cite{kalachev_two-sided_2023} for our case of planted inner codes where $\allones\in C_A,C_B^\perp$. As these inner codes are constant-sized as $n\rightarrow\infty$, the randomized construction can be made strongly explicit by a brute force search.

An interesting consequence of our result is that we can construct planted quantum Tanner codes $\cC$ of positive dimension $k>0$ with inner codes $C_A,C_B$ of any desired respective rates $R_A,R_B\in(0,1)$; for instance, we can take $R_A=R_B$. In contrast, the prior technique of bounding $k$ by counting parity checks only implies that $k\geq-(1-2R_A)(1-2R_B)\cdot n$, which never gives a meaningful bound when $R_A=R_B$. Thus our construction allows instantiations in new parameter regimes.

We also remark that while we only show how to plant a nontrivial codeword in the qLDPC codes of \cite{leverrier_quantum_2022-1}, our techniques also apply to the codes of \cite{panteleev_asymptotically_2022}; to avoid redundancy we do not spell out the details.

\subsection{Strongly Explicit SoS Lower Bounds}
\label{sec:sosinf}

The Sum-of-Squares semidefinite programming hierarchy is one of the most powerful algorithmic frameworks for approximating the satisfiability of CSPs (see \cite{fleming_semialgebraic_2019} for a survey). However, almost all of the known hard instances (i.e.~lower bounds) for this hierachy are given by randomized constructions. Hopkins and Lin \cite{hopkins_explicit_2022-1}, building on the techniques of Dinur et al.~\cite{dinur_explicit_2021}, constructed the first explicit unsatisfiable CSPs that cannot be refuted by a linear number of levels of the SoS SDP hierarchy. In contrast, explicit lower bounds prior to their work applied to at best a logarithmic number of levels of the SoS hierarchy.

Hopkins and Lin \cite{hopkins_explicit_2022-1} proved their result by showing that hard instances for SoS can be obtained from a family of qLDPC codes with small-set boundary and coboundary expansion. Explicit such qLDPC codes, such as the quantum Tanner codes of \cite{leverrier_quantum_2022-1}, then yield the desired explicit hard CSPs.

\begin{remark}
  The SoS lower bounds of \cite{hopkins_explicit_2022-1} marked the first complexity theoretic application of linear-distance qLDPC codes; the subsequent proof of the NLTS theorem \cite{anshu_nlts_2023} provided a second notable application. Such applications were perhaps surprising given that the construction of asymptotically good qLDPC codes, first obtained by \cite{panteleev_asymptotically_2022} and subsequently extended and modified by \cite{leverrier_quantum_2022-1,dinur_good_2023}, was originally motivated in large part by applications to quantum error correction.
\end{remark}

However, the explicitness of the CSP construction in \cite{hopkins_explicit_2022-1} was weak in the sense of Definition~\ref{def:weakvstrong} below. One of the major questions left open by their work was to make this construction strongly explicit \cite{hopkins_personal_2023}. We apply our construction of planted quantum Tanner codes in Theorem~\ref{thm:plantedresultinf} to resolve this problem.

\begin{definition}[Weak vs.~strong explicitness]
  \label{def:weakvstrong}
  Let $X=(x_n)_{n\in\bN}$ be an infinite family of objects such that each $x_n$ can be represented by a bitstring $x_n\in\{0,1\}^{a_n}$ of length $a_n$, where $a_n\rightarrow\infty$ as $n\rightarrow\infty$. We say that $X$ is:
  \begin{itemize}
  \item \textbf{weakly explicit} (or simply ``explicit'') if there exist a $\poly(a_n)$-time algorithm $A(n)$ that outputs $x_n$
  \item \textbf{strongly explicit} if there exists a $\poly(\log n,\log a_n)$-time algorithm $A(n,i)$ that outputs the $i$th bit of $x_n$ for $i\in[a_n]$.
  \end{itemize}
\end{definition}

We specifically say that a family of matrices is weakly (resp.~strongly) explicit if for each $n\times m$ matrix in the family, there is a $\poly(n,m)$ (resp.~$\poly(\log n,\log m)$) time algorithm to compute the $j$th nonzero entry of the $i$th row, as well as the $j$th nonzero entry of the $i$th column.

Then a family of graphs is weakly (resp.~strongly) explicit if the associated adjacency matrices are weakly (resp.~strongly) explicit. Similarly, a CSS code $\cC=\CSS(C_X=\ker H_X,C_Z=\ker H_Z)$ is weakly (resp.~strongly) explicit if the matrices $H_X,H_Z$ are weakly (resp.~strongly) explicit. 

As another relevant example, consider a family of CSPs given by $\ell$-LIN instances, which are defined by $n$ linear constraints on $m$ variables over a fixed finite field $\bF_q$, such that each linear equation has $\leq\ell=O(1)$ nonzero coefficients. A family of such $\ell$-LIN instances is weakly (resp.~strongly) explicit if the $i$th linear equation can be computed in time $\poly(n,m)$ (resp.~$\poly(\log n,\log m)$).

Given a qLDPC code $\cC=\CSS(C_X=\ker H_X,C_Z=\ker H_Z)$ of locality $\ell$ and an arbitrary element $\beta\in C_X\setminus C_Z^\perp$, Hopkins and Lin~\cite{hopkins_explicit_2022-1} considered the associated $\ell$-LIN instance $\cI_{\cC,\beta}$ with $m=m_Z$ variables $y_1,\dots,y_m$ and $n$ linear constraints over $\bF_q$ given by the system of equations $H_Z^\top y=\beta$ for $y=(y_1,\dots,y_m)$. They showed that if $\cC$ has $(\Omega(1),\Omega(1))$-small-set boundary and coboundary expansion, then at most $1-\Omega(1)$ fraction of the constraints in $\cI_{\cC,\beta}$ can be satisfied, but $\cI_{\cC,\beta}$ is hard to refute for $\Omega(n)$ levels of SoS. Furthermore, they presented a reduction to reduce the size of the constraints, thereby providing similarly unsatisfiable and hard instances of 3-LIN over $\bF_2$, that is, of 3-XOR.

However, even if $\cC$ comes from a strongly explicit family of qLDPC codes, the associated $\ell$-LIN instance $\cI_{\cC,\beta}$ is only weakly explicit in general, as one must perform Gaussian elimination to compute some $\beta\in C_X\setminus C_Z^\perp$, which takes $\poly(n,m)$ time.

Because our planted quantum Tanner codes in Theorem~\ref{thm:plantedresultinf} by construction have $\allones\in C_X\setminus C_Z^\perp$, they resolve this issue, and hence yield the following result.

\begin{theorem}[Strongly explicit SoS lower bounds; restatement of Corollary~\ref{cor:sosresult} and Corollary~\ref{cor:3xor}]
  \label{thm:sosresultinf}
  The $\ell$-LIN instances $\cI_{\cC,\allones}$ for planted quantum Tanner codes $\cC$ over any fixed prime-sized alphabet $\bF_p$ provide a family of strongly explicit instances such that each $\cI_{\cC,\allones}$:
  \begin{enumerate}
  \item has $\Theta(n)$ variables and constraints,
  \item has satisfiability $\leq(1-\Omega(1))$,
  \item cannot be refuted by $cn$ levels of the SoS hierarchy for a sufficiently small constant $c>0$.
  \end{enumerate}
  Furthermore, there exists a strongly explicit family of 3-XOR (i.e.~3-LIN over $\bF_2$) instances that also satisfies the three properties above.
\end{theorem}

We remark that \cite{hopkins_explicit_2022-1} actually restricted attention to the binary alphabet $q=2$ case, though their SoS lower bounds for $\ell$-LIN extend to larger prime alphabets. We suspect that their reduction to 3-LIN similarly extends to larger alphabets, though for conciseness we do not check the details.


\subsection{Open Questions}
\label{sec:openquestions}
Our results raise the following open questions:
\begin{itemize}
\item Can our construction of NLTS Hamiltonians from low-rate qLDPC codes lead to more progress towards qPCP or hardness of approximation? For instance, perhaps the fact that low-rate codes, which correspond to Hamiltonians with low-dimensional ground spaces, suffice for NLTS will be helpful in constructing Hamiltonians with stronger hardness of approximation guarantees.
\item Our results highlight the usefulness of low-rate qLDPC codes, and suggest that for complexity theoretic applications there is often little benefit to having high rate. However, to the best of our knowledge, our planted quantum Tanner codes provide the only known ``inherently'' low-rate qLDPC codes, and they still have high rate in some parameter regimes. In contrast, there are many interesting classical low-rate LDPC codes such as Hadamard and Reed-Muller codes, which have properties not shared by any high-rate codes. In the quantum case, can similar stronger properties be obtained by allowing for low rate in qLDPC codes?
\end{itemize}

\section{Notation}
\label{sec:notation}
For a string $y\in\bF_q^n$, we denote the Hamming weight by $|y|=|\{i\in[n]:y_i\neq 0\}|$. For subsets $S,T\subseteq\bF_q^n$, we denote the Hamming distance by $\dis(S,T)=\min_{s\in S,t\in T}|s-t|$.

Unless explicitly stated otherwise, by a ``code'' we mean a linear subspace $C\subseteq\bF_q^n$. The code $C$ has block length $n$, dimension $k=\dim_{\bF_q}(C)$, and distance $d=\min_{y\in C\setminus\{0\}}|y|$, which can be summarized by saying it is a $[n,k,d]_q$ code. The dual code is $C^\perp=\{x\in\bF_q^n:x\cdot y=0\;\forall y\in C\}$.

For codes $C_i\subseteq\bF_q^{n_i}$ for $i=1,2$, the tensor code $C_1\otimes C_2\subseteq\bF_q^{n_1\times n_2}$ consists of all $n_1\times n_2$ matrices where every column lies in $C_1$ and every row lies in $C_2$. The dual of the tensor code is $(C_1\otimes C_2)^\perp=C_1^\perp\otimes\bF_q^{n_2}+\bF_q^{n_1}\otimes C_2^\perp$.

Given a $\Delta$-regular graph $\Gamma$ with $n$ edges and an inner code $C_{\text{in}}\subseteq\bF_q^\Delta$, we denote the associated classical Tanner code by $C=\Tan(\Gamma,C_{\text{in}})\subseteq\bF_q^n$, which is constructed as follows. We associate the set of all edges in $\Gamma$ with the set $[n]$, and we associate the set of edges incident to each vertex $v$ in $\Gamma$ with the set $[\Delta]$. Then we define $C$ to be the set of all edge labelings $y\in\bF_q^n=\bF_q^{E(\Gamma)}$ such that the labels of edges incident to each $v\in\Gamma$ form a codeword in $C_{\text{in}}$.

For a pure quantum state $\ket{\psi}$, we denote the density matrix by $\psi=\ket{\psi}\bra{\psi}$. For a set $S\subseteq\bF_q^n$, we let $\ket{S}=|S|^{-1/2}\sum_{s\in S}\ket{s}$ denote the uniform superposition over elements of $S$.

The quantum codes we consider in this paper are CSS codes, which are defined as follows. For classical codes $C_X,C_Z\subseteq\bF_q^n$ such that $C_X^\perp\subseteq C_Z$, the associated quantum CSS code $\cC=\CSS(C_X,C_Z)$ is defined by $\cC=\spn\{\ket{y+C_X^\perp}:y\in C_Z\}\subseteq(\bC^q)^{\otimes n}$. This code has block length $n$, dimension $k=\log_q\dim_{\bC}(\cC)=\dim_{\bF_q}(C_Z)-\dim_{\bF_q}(C_X^\perp)$, and distance $d=\min_{y\in (C_Z\setminus C_X^\perp)\cup(C_X\setminus C_Z^\perp)}|y|$, which can be summarized by saying that $\cC$ is a $[[n,k,d]]_q$ code.

If $C_X=\ker H_X$ and $C_Z=\ker H_Z$ for parity check matrices $H_X,H_Z$ in which each row and column has Hamming weight $\leq\ell$, we say that $\cC$ is a CSS code with check weight, or locality, $\leq\ell$. A family of codes with constant locality $\ell=O(1)$ as $n\rightarrow\infty$ is said to be LDPC. The family of codes is (strongly) explicit if the associated families of parity check matrices $H_X,H_Z$ are (strongly) explicit.

\section{Review of Quantum Tanner Codes}
\label{sec:qtan}
In this section we review the construction and relevant properties of the asymptotically good quantum LDPC codes of Leverrier and Z\'{e}mor \cite{leverrier_quantum_2022-1,leverrier_decoding_2023}, which are called quantum Tanner codes. Although \cite{leverrier_quantum_2022-1,leverrier_decoding_2023} present the construction over binary alphabets, we consider arbitrary finite field alphabets; all their results and proofs extend to this more general case with just some `$+$' signs changed to `$-$' signs for fields of characteristic $\neq 2$.

Recall that a classical Tanner code is constructed by imposing constraints from an inner code on a graph (see Section~\ref{sec:notation}). In contrast, a quantum Tanner code $\cC$ is constructed by imposing constraints from \textit{two} inner codes on a higher-dimensional object called a \textit{square Cayley complex}. In particular, $\cC=\CSS(C_X,C_Z)$, where both $C_X,C_Z$ are classical Tanner codes on graphs obtained from a square Cayley complex, with distinct inner codes.

\subsection{Construction}
\label{sec:qTanconstruct}
We now desribe the construction of a quantum Tanner code $\cC=\CSS(C_X,C_Z)$. We first need to define a square Cayley complex. Recall that for a group $G$ and a subset $A\subseteq G$, the Cayley graph $\Cay(G,A)$ has vertex set $G$ and edge set $\{(g,ag):g\in G,\;a\in A\}$. As described below, a square Cayley complex is a sort of 2-dimensional generalization of a Cayley graph.

A square Cayley complex consists a tuple $(V,E,Q)$ of vertices, edges, and faces (or ``squares'') that is specified by a group $G$ and two generating sets $A,B\subseteq G$ as follows. We typically take $|A|=|B|=\Delta=O(1)$ as $|G|=\Theta(n)\rightarrow\infty$, and assume that $A=A^{-1}$ and $B=B^{-1}$ are closed under inversion. The complex then has vertex set $V=G\times\{0,1\}^2$, edge set $E=E_A\sqcup E_B$ for
\begin{align*}
  E_A &= \{(g,i0),(ag,i1):g\in G,\;i\in\{0,1\},\;a\in A\} \\
  E_B &= \{(g,0j),(gb,1j):g\in G,\;j\in\{0,1\},\;b\in B\},
\end{align*}
and face set
\begin{align*}
  Q &= \{(g,00),(ag,01),(gb,10),(agb,11):g\in G,\;a\in A,\;b\in B\}.
\end{align*}

For $i,j\in\{0,1\}$, let $V_{ij}=G\times(i,j)$. Define bipartite graphs $\Gamma_0=(V_{00}\sqcup V_{11},Q)$ and $\Gamma_1=(V_{01}\sqcup V_{10},Q)$ whose edges are given by pairs of vertices that form a diagonal in a square in $Q$; for instance, $\Gamma_0$ has an edge between $v\in V_{00}$ and $v'\in V_{11}$ if $v,v'$ share a face in $Q$. Observe that both $\Gamma_0$ and $\Gamma_1$ have a unique edge associated to each square in $Q$. Furthermore, the $\Gamma_i$-edges incident to a vertex $v$ correspond to the squares in $Q$ that contain $v$; we let $Q(v)$ denote the set of these squares. But by definition $Q(v)$ consists of an $A\times B$ grid of squares. For instance, for $v=(g,00)\in V_{00}$ then
\begin{equation*}
  Q(v) = \{(g,00),(ag,01),(gb,10),(agb,11):a\in A,\;b\in B\}.
\end{equation*}

Therefore given a square Cayley complex $(V,E,Q)$ of degree $\Delta=|A|=|B|$ along with classical codes $C_A\subseteq\bF_q^A=\bF_q^\Delta$ and $C_B\subseteq\bF_q^B=\bF_q^\Delta$, we may define a quantum Tanner code $\cC=\CSS(C_X,C_Z)$ by
\begin{align*}
  C_X &= \Tan(\Gamma_0,(C_A\otimes C_B)^\perp) \\
  C_Z &= \Tan(\Gamma_1,(C_A^\perp\otimes C_B^\perp)^\perp).
\end{align*}
That is, $C_X$ and $C_Z$ are classical Tanner codes on the graphs $\Gamma_0$ and $\Gamma_1$ respectively, where the inner codes are given by dual tensor codes. Because $E(\Gamma_0)\cong E(\Gamma_1)\cong Q$, both $C_X$ and $C_Z$ are subspaces of $\bF_q^Q$.

\subsection{Locality and Dimension}
\label{sec:locdim}
We now describe some basic properties of $\cC$. By definition $\cC$ has block length $n=|Q|$. Parity checks for $C_X$ are given by tensor codewords in $C_A\otimes C_B$ supported in the neighborhood $Q(v_0)$ of any $v_0\in V_{00}\sqcup V_{11}$. Similarly, parity checks for $C_Z$ are given by tensor codewords in $C_A^\perp\otimes C_B^\perp$ supported in the neighborhood $Q(v_1)$ of any $v_1\in V_{01}\sqcup V_{10}$. Because any such $Q(v_0),Q(v_1)$ are either disjoint or intersect in a single row or column, the parity checks for $C_X$ and $C_Z$ are orthogonal, so $C_X^\perp\subseteq C_Z$. Furthermore, as $|Q(v_0)|=|Q(v_1)|=\Delta^2$, the quantum Tanner code $\cC$ is LDPC with locality $\Delta^2=O(1)$ as $n=|Q|\rightarrow\infty$.

Counting parity checks to bound the number of linear constraints on $C_X,C_Z$ implies that $\cC$ has dimension $k\geq-(1-2R_A)(1-2R_B)\cdot n$, where $R_A=\dim(C_A)/\Delta$ and $R_B=\dim(C_B)/\Delta$ denote the rate of $C_A$ and $C_B$ respectively.

\subsection{Distance}
\label{sec:dis}
To present the distance bound for quantum Tanner codes, we need the following definition.

\begin{definition}[Product-expansion]
  \label{def:prodexp}
  A pair of codes $C_1,C_2\subseteq\bF_q^n$ is \textbf{$\rho$-product-expanding} if every $x\in(C_1^\perp\otimes C_2^\perp)^\perp=C_1\otimes\bF_q^n+\bF_q^n\otimes C_2$ can be decomposed as $x=c+r$ for some $c\in C_1\otimes\bF_q^n$ and $r\in\bF_q^n\otimes C_2$ satisfying
  \begin{equation*}
    |x| \geq \rho n(|c|_{\text{col}}+|r|_{\text{row}}),
  \end{equation*}
  where $|c|_{\text{col}}$ denotes the number of nonzero columns in $c$ and $|r|_{\text{row}}$ denotes the number of nonzero rows in $r$.
\end{definition}

It is immediate that product-expansion yields a bound on the distances of the associated codes:

\begin{lemma}[Well known]
  \label{lem:exptodis}
  If the pair $C_1,C_2\subseteq\bF_q^n$ is $\rho$-product expanding, then $C_1$ and $C_2$ have distance $\geq\rho n$.
\end{lemma}
\begin{proof}
  Let $x\in C_1\otimes\bF_q^n$ have its first column be a minimum-weight nonzero codeword of $C_1$, and have all other columns be $0$. Then $\rho$-product-expansion implies that $C_1$ has distance $|x|\geq\rho n$. A similar argument holds for $C_2$.
\end{proof}

The following result bounding the product expansion of random pairs of codes was shown independently by \cite{kalachev_two-sided_2023} and \cite{dinur_good_2023}, though only the former explicitly considered non-binary alphabets.

\begin{proposition}[\cite{kalachev_two-sided_2023}]
  \label{prop:randomprod}
  Fix any finite field $\bF_q$. For every fixed $\epsilon>0$, there exists a constant $\rho=\rho(\epsilon)>0$ and a function $\delta(n)=\delta(n;\epsilon)\rightarrow 0$ as $n\rightarrow\infty$ such that the following holds. For every pair of integers $k_1,k_2\in(\epsilon n,(1-\epsilon)n)$, if $C_i\subseteq\bF_q^n$ for $i=1,2$ is drawn uniformly at random from the set of linear codes of dimension $k_i$, then with probability $\geq 1-\delta(n)$ the pair $(C_1,C_2)$ will be $\rho$-product-expanding.
\end{proposition}

Applying Proposition~\ref{prop:randomprod} with a union bound over $(C_1,C_2)$ and $(C_1^\perp,C_2^\perp)$ immediately yields the following corollary.

\begin{corollary}[\cite{kalachev_two-sided_2023}]
  \label{cor:randomprod}
  Defining all variables as in Proposition~\ref{prop:randomprod}, then with probability $\geq 1-2\delta(n)$ both $(C_1,C_2)$ and $(C_1^\perp,C_2^\perp)$ will be $\rho$-product-expanding.
\end{corollary}

The distance bound for quantum Tanner codes will also rely on the Cayley graphs $\Cay(G,A)$ and $\Cay(G,B)$ having sufficiently good expansion.

\begin{definition}
  For a regular graph $\Gamma$ of degree $\Delta(\Gamma)$, the \textbf{(unnormalized) spectral expansion} $\lambda(\Gamma)$ is the second largest absolute value of an eigenvalue of the adjacency matrix of $\Gamma$. If $\lambda(\Gamma)\leq 2\sqrt{\Delta(\Gamma)-1}$, then $\Gamma$ is \textbf{Ramanujan}. Meanwhile, if an infinite family of regular graphs $\Gamma$ all satisfy $\lambda(\Gamma)\leq\Delta(\Gamma)^{1/2+o(1)}$, then the family is \textbf{almost Ramanujan}. Here $o(1)$ denotes any function of $\Delta$ that approaches $0$ as $\Delta\rightarrow\infty$.
\end{definition}

Constructions of Ramanujan Cayley graphs have for instance been given by \cite{lubotzky_ramanujan_1988} and \cite{morgenstern_existence_1994}; the latter construction in particular is strongly explicit:

\begin{theorem}[\cite{morgenstern_existence_1994}]
  \label{thm:morgram}
  For every prime power $q\geq 3$, there exists a strongly explicit family of $(q+1)$-regular Ramanujan Cayley graphs $(\Gamma_m)_{m\in\bN}$ with the number of vertices given by
  \begin{equation*}
    |V(\Gamma_m)| = \begin{cases}
      q^{2m}(q^{4m}-1),&q\equiv 0\pmod{2}\\
      q^{2m}(q^{4m}-1)/2,&q\equiv 1\pmod{2}.\\
    \end{cases}
  \end{equation*}
\end{theorem}

The graphs in Theorem~\ref{thm:morgram} can be used to instantiate strongly explicit linear-distance quantum Tanner codes by \cite{leverrier_quantum_2022-1,leverrier_decoding_2023}. However, they are not quite sufficient for our purposes. Specifically, using these graphs, our planted quantum Tanner codes and the resulting strongly explicit SoS lower bounds described in Section~\ref{sec:plant} would hold only for alphabets $\bF_q$ of characteristic $p\geq 7$. This restricton on the field size arises because the graphs in Theorem~\ref{thm:morgram} have $|V(\Gamma_m)|$ divisible by $2$, $3$, and $5$, but our results require Cayley expanders $\Gamma$ for which $q$ is relatively prime with $|V(\Gamma)|$.

We therefore instead use the Cayley expanders given by Example~3.4 in \cite{lubotzky_groups_1993}, for which the number of vertices is guaranteed to be a power of any desired prime. \cite{lubotzky_groups_1993} showed that these graphs have constant degree $\Delta$ and constant spectral expansion $<\Delta$. By amplifying the expansion to be near-Ramanujan by applying Theorem~1.2 in \cite{jeronimo_almost_2022}, we obtain the following result, which we formally prove in Section~\ref{sec:expconst}.

\begin{restatable}[Follows from \cite{lubotzky_groups_1993,jeronimo_almost_2022}]{theorem}{lwexp}
  \label{thm:lwexp}
  For every prime $p$, there is an infinite set $\mathbf{\Delta}\subseteq\bN$ for which there exists a strongly explicit family of almost-Ramanujan Cayley graphs $(\Gamma_{m,\Delta})_{m\in\bN,\Delta\in\mathbf{\Delta}}$, where $\Gamma_{m,\Delta}$ has $|V(\Gamma_{m,\Delta})|=p^{3m}$ vertices and has degree $\Delta$. Furthermore, we may choose $\mathbf{\Delta}$ such that for every $\Delta\in\mathbf{\Delta}$, either $\Delta+1\in\mathbf{\Delta}$ or $\Delta-1\in\mathbf{\Delta}$.
\end{restatable}


We are now ready the present the distance bound for quantum Tanner codes.

\begin{theorem}[\cite{leverrier_quantum_2022-1,leverrier_decoding_2023}]
  \label{thm:qTandis}
  For every fixed $\rho>0$, the following holds for all sufficiently large $\Delta$. Let $\cC$ be a quantum Tanner code for which:
  \begin{enumerate}
  \item\label{it:ramanujan} $\Cay(G,A),\Cay(G,B)$ are almost-Ramanujan graphs of degree $\Delta$.
  \item\label{it:innerprod} $(C_A,C_B),(C_A^\perp,C_B^\perp)$ are $\rho$-product-expanding.
  \end{enumerate}
  Then $\cC$ has distance $d\geq cn$ for a constant $c>0$ depending only on $\rho,\Delta$.
\end{theorem}

Recall that by Lemma~\ref{lem:exptodis}, Condition~\ref{it:innerprod} in Theorem~\ref{thm:qTandis} implies that $C_A,C_B,C_A^\perp,C_B^\perp$ have distance $\geq\rho\Delta$.

Condition~\ref{it:ramanujan} in Theorem~\ref{thm:qTandis} can be met using the strongly explicit Ramanujan graphs in Theorem~\ref{thm:morgram}, or using the strongly explicit almost-Ramanujan graphs in Theorem~\ref{thm:lwexp}. If $C_A,C_B\subseteq\bF_q^\Delta$ are chosen to be random codes of some fixed rates $0<R_A,R_B<1$ for any sufficiently large constant $\Delta$, then
Condition~\ref{it:innerprod} is met by Corollary~\ref{cor:randomprod}. Because $\Delta$ is a constant as $n\rightarrow\infty$, we may find $C_A,C_B$ in constant time by a brute force search, so the overall construction of $\cC$ is strongly explicit.

While \cite{leverrier_quantum_2022-1,leverrier_decoding_2023} only proved Theorem~\ref{thm:qTandis} in the case where $\Cay(G,A),\Cay(G,A)$ are Ramanujan graphs of degree $\Delta$, their proof generalizes flawlessly to allow for almost-Ramanujan graphs. Specifically, the proof of linear distance in Section~3 of \cite{leverrier_decoding_2023} defines ``exceptional vertices'' using a parameter $\alpha=\delta^2/256$, where $\delta$ denotes the relative distance of the inner codes. If we for instance redefine $\alpha=\delta^2/\Delta^{1/100}$, we can essentially carry through their exact same analysis assuming $\Cay(G,A),\Cay(G,B)$ are almost-Ramanujan. An additional factor of $\Delta^{1/50+o(1)}$ arises in the bound in their Lemma~10, and an additional $\Delta^{o(1)}$ factor arises in the bound in their Lemma~11, but the proof of linear still goes through with these slightly worse parameters.

We suspect that almost-Ramanujan expansion $\lambda\leq\Delta^{1/2+o(1)}$ is still stronger than necessary, and it may in fact be sufficient to have spectral expansion $\lambda$ equal to some small constant fraction of $\Delta$. However, it is less transparent how to modify the proof of \cite{leverrier_decoding_2023} for this case, so for conciseness we do not pursue this direction.

\subsection{Small-Set (Co)boundary Expansion}
\label{sec:ssexp}
For our applications of quantum Tanner codes, we will need a stronger notion than distance, called \textit{small-set (co)boundary expansion}, which was first formally stated in the context of quantum codes by Hopkins and Lin \cite{hopkins_explicit_2022-1}. Below, for a code $C$, we denote $|y|_C=\min_{y'\in C}|y+y'|$.

\begin{definition}[Small-set (co)boundary expansion]
  \label{def:ssexp}
  Let $\cC=\CSS(C_X=\ker H_X,C_Z=\ker H_Z)$ be a CSS code given by parity check matrices $H_X\in\bF_q^{m_X\times n}$ and $H_Z\in\bF_q^{m_Z\times n}$. For $c_1,c_2>0$, we say that $\cC$ has \textbf{$(c_1,c_2)$-small-set boundary expansion} if it holds for every $y\in\bF_q^n$ with $|y|\leq c_1n$ that
  \begin{equation*}
    \frac{|H_Zy|}{m_Z} \geq c_2\frac{|y|_{C_X^\perp}}{n}.
  \end{equation*}
  Similarly, $\cC$ has \textbf{$(c_1,c_2)$-small-set coboundary expansion} if it holds for every $y\in\bF_q^n$ with $|y|\leq c_1n$ that
  \begin{equation*}
    \frac{|H_Xy|}{m_X} \geq c_2\frac{|y|_{C_Z^\perp}}{n}.
  \end{equation*}
\end{definition}

Small-set (co)boundary expansion immediately implies a bound on the distance of the code:

\begin{lemma}[Well known]
  \label{lem:ssexptodis}
  If $\cC=\CSS(C_X,C_Z)$ of block length $n$ has $(c_1,c_2)$-small set boundary and coboundary expansion for $c_1,c_2>0$, then $\cC$ has distance $\geq c_1n$.
\end{lemma}
\begin{proof}
  For every $y\in \bF_q^n\setminus C_X^\perp$ with $|y|\leq c_1n$, then $|y|_{C_X^\perp}>0$, so small-set boundary expansion implies that $|H_Zy|\geq c_2m_Z|y|_{C_X^\perp}/n>0$ and thus $y\notin C_Z$. An analogous argument shows that $C_X\setminus C_Z^\perp$ has no elements of weight $\leq c_1n$.
\end{proof}

It was originally observed that quantum Tanner codes have small set (co)boundary expansion in \cite{hopkins_explicit_2022-1}. We remark that another proof is given implicitly by the decoder of Leverrier and Z\'{e}mor \cite{leverrier_decoding_2023}. Specifically, there exists a constant $c_1>0$ such that for any errors $e_X,e_Z\in\bF_q^n$ of sufficiently low weight $|e_X|,|e_Z|\leq c_1n$, the decoder of \cite{leverrier_decoding_2023} takes as input the syndromes $s_X=H_Xe_X$ and $s_Z=H_Ze_Z$, and outputs some $e_X'\in e_X+C_Z^\perp$ and $e_Z'\in e_Z+C_X^\perp$ such that $|e_X'|\leq O(|s_X|),\;|e_Z'|\leq O(|s_Z|)$. It follows that $|e_X|_{C_Z^\perp}\leq|e_X'|\leq O(|s_X|)$ and $|e_Z|_{C_X^\perp}\leq|e_Z'|\leq O(|s_Z|)$, which are precisely the conditions required by small-set coboundary and boundary expansion, respectively. The result below formally summarizes this small-set (co)boundary expansion.

\begin{theorem}[\cite{hopkins_explicit_2022-1,leverrier_decoding_2023}]
  \label{thm:qTanexp}
  For every fixed $\rho>0$, the following holds for all sufficiently large $\Delta$. Let $\cC$ be a quantum Tanner code that satisfies Conditions~\ref{it:ramanujan}, \ref{it:innerprod} in the statement of Theorem~\ref{thm:qTandis}. Then $\cC$ has $(c_1,c_2)$-small-set boundary and coboundary expansion for constants $c_1,c_2>0$ depending only on the values of $\rho,\Delta$.
\end{theorem}

Note that Theorem~\ref{thm:qTanexp} implies Theorem~\ref{thm:qTandis} by Lemma~\ref{lem:ssexptodis}.

As described in Section~\ref{sec:dis}, \cite{leverrier_decoding_2023} assume the Cayley graphs $\Cay(G,A),\Cay(G,B)$ are Ramanujan, whereas we make the slightly weaker assumption that they are almost-Ramanujan. However, the decoding proof in \cite{leverrier_decoding_2023} is similar to the distance proof, and the extension to almost-Ramanujanness is nearly identical. Specifically, while Section~5 of \cite{leverrier_decoding_2023} defines a parameter $\alpha=\delta^2\epsilon^2/2^{10}$, we redefine this parameter to be $\delta^2\epsilon^2/\Delta^{1/100}$, and carry through the rest of the proof essentially as before. An additional factor of $\Delta^{1/50+o(1)}$ arises in the bound in their Lemma~15, and an additional factor of $\Delta^{o(1)}$ arises in the bound in their Lemma~16, but otherwise the decoding analysis goes through as before, with these slightly worse parameters.

Leverrier and Z\'{e}mor \cite{leverrier_efficient_2023} show how their decoder can also be used to decode the asymptotically good qLDPC codes of Panteleev and Kalachev \cite{panteleev_asymptotically_2022}; again in this case the decoder outputs an error whose weight is linear in the syndrome weight. Therefore a similar result as Theorem~\ref{thm:qTanexp} holds for the codes of \cite{panteleev_asymptotically_2022} as well.

\section{NLTS Hamiltonians from Codes of Arbitrary Dimension}
\label{sec:nlts}
In this section, we show that quantum LDPC codes with linear distance and an appropriate clustering property yield NLTS Hamiltonians, regardless of the code dimension. This result improves upon the prior construction of NLTS Hamiltonians of \cite{anshu_nlts_2023}, which required the stronger assumption that the code dimension be linearly large. For simplicity in this section, we restrict attention to binary alphabets, though we expect the results to generalize naturally to qudits for more general alphabet sizes.

\subsection{Setup of the Local Hamiltonian}
We let $\cC=\CSS(C_X,C_Z)=\spn\{\ket{y+C_X^\perp}:y\in C_Z\}$ be an $[[n,k,d]]_2$ quantum LDPC CSS code, where all parity checks have weight $\leq\ell$. We assume $\cC$ belongs to a family of such codes with constant relative distance $d/n=\Omega(1)$ and constant locality $\ell=O(1)$ as the block length $n\rightarrow\infty$. We will also assume that $\cC$ satisfies the clustering property described in Definition~\ref{def:clustering} below. The main novel aspect of our proof is that it holds for any nonzero code dimension $k>0$. In contrast, the prior NLTS proof \cite{anshu_nlts_2023} assumed the rate $k/n$ is constant as $n\rightarrow\infty$.

Recent good qLDPC codes, such as the quantum Tanner codes of \cite{leverrier_quantum_2022-1}, satisfy all of the conditions above, and have constant rate. Our planted quantum Tanner codes in Theorem~\ref{thm:plantedresult} also satisfy these conditions, but in some parameter regimes have dimension 1 (i.e.~inverse linear rate). Hence our planted quantum Tanner codes provide examples of qLDPC codes of subconstant rate where our NLTS result applies.

Denote the parity check matrices of $C_X$ and $C_Z$ by $H_X\in\bF_2^{m_X\times n}$ and $H_Z\in\bF_2^{m_Z\times n}$ respectively, so that $C_X=\ker H_X$ and $C_Z=\ker H_Z$. By assumption all rows of $H_X,H_Z$ have $\leq\ell$ nonzero entries. Also define $G_X^\epsilon=\{y\in\bF_2^n:|H_Xy|\leq\epsilon m_Z\}$, and define $G_Z^\epsilon$ analogously. We assume $\cC$ satisfies the following clustering property for $G_X^\epsilon$ and $G_Z^\epsilon$, which is stated as Property~1 in \cite{anshu_nlts_2023}. Below, we denote $|y|_C=\min_{y'\in C}|y+y'|$.

\begin{definition}[Clustering property \cite{anshu_nlts_2023}]
  \label{def:clustering}
  For constants $c_1,c_2,\epsilon_0>0$, we say that $\cC=\CSS(C_X,C_Z)$ exhibits \textbf{$(c_1,c_2,\epsilon_0)$-clustering} if for all $0<\epsilon<\epsilon_0$, the following hold:
  \begin{enumerate}
  \item Every $y\in G_X^\epsilon$ satisfies either $|y|_{C_Z^\perp}\leq c_1\epsilon n$ or $|y|_{C_Z^\perp}\geq c_2n$.
  \item Every $y\in G_Z^\epsilon$ satisfies either $|y|_{C_X^\perp}\leq c_1\epsilon n$ or $|y|_{C_X^\perp}\geq c_2n$.
  \end{enumerate}
\end{definition}

This clustering property follows from small-set (co)boundary expansion (Definition~\ref{def:ssexp}), as is shown below.

\begin{lemma}
  If $\cC$ has $(c_1',c_2')$-small-set boundary and coboundary expansion, then $\cC$ has $(c_1,c_2,\epsilon_0)$-clustering for $c_1=1/c_2'$, $c_2=c_1'$, $\epsilon_0=1$.
\end{lemma}
\begin{proof}
  Assume that $y\in G_Z^\epsilon$ satisfies $|y|_{C_X^\perp}\leq c_2n=c_1'n$. Let $y'$ be the minimum-weight element of $y+C_X^\perp$, so that $H_Zy'=H_Zy$ and $|y'|=|y|_{C_X^\perp}$. Then the small-set boundary expansion implies that $|H_Zy'|/m_Z\geq c_2'\cdot|y'|/n$, so $|y|_{C_X^\perp}=|y'|\leq(n/c_2'm_Z)|H_Zy'|\leq \epsilon n/c_2'=c_1\epsilon n$. Thus we have shown the desired clustering for $G_Z^\epsilon$; an analogous argument applies to $G_X^\epsilon$.
\end{proof}

For the remainder of Section~\ref{sec:nlts}, we assume that $\cC$ satisfies $(c_1,c_2,\epsilon_0)$-clustering for some constants $c_1,c_2,\epsilon_0>0$ as $n\rightarrow\infty$.

Following \cite{anshu_nlts_2023}, we define our $\ell$-local Hamiltonian $\Ham$ to be the code Hamiltonian
\begin{equation*}
  \Ham = \frac12(\Ham_X+\Ham_Z)
\end{equation*}
for
\begin{align*}
  \Ham_X &= \frac{1}{m_X}\sum_{y\in\mathrm{rows}(H_X)}\frac{I-X^y}{2} \\
  \Ham_Z &= \frac{1}{m_Z}\sum_{y\in\mathrm{rows}(H_Z)}\frac{I-Z^y}{2}.
\end{align*}
Thus in particular the ground space of $\Ham$ is precisely the code space $\cC=\spn\{\ket{y+C_X^\perp}:y\in C_Z\}$

While our general proof will follow that of \cite{anshu_nlts_2023}, our use of an uncertainty principle instead follows the earlier work of \cite{eldar_local_2017}. Below we state the uncertainty principle we will use, which appears implicitly in \cite{eldar_local_2017}.

\begin{lemma}[Uncertainty principle \cite{hofmann_violation_2003,eldar_local_2017}]
  \label{lem:uncertainty}
  Let $A,B$ be Hermitian observables with $AB+BA=0$ and $A^2=B^2=I$. Then for every (possibly mixed) state $\rho$, at least one of the inequalities $|\Tr(A\rho)|\leq 1/2+1/2\sqrt{2}$ or $|\Tr(B\rho)|\leq 1/2+1/2\sqrt{2}$ holds.
\end{lemma}

For completeness, we present a proof of Lemma~\ref{lem:uncertainty} from the following result for pure states, which is given as Lemma~37 in \cite{eldar_local_2017}, but previously shown by \cite{hofmann_violation_2003}.

\begin{lemma}[\cite{hofmann_violation_2003}]
  \label{lem:uncertaintypure}
  Let $A,B$ be Hermitian observables with $AB+BA=0$ and $A^2=B^2=I$, and let $\ket{\psi}$ be a pure state. Letting $\Delta A^2=\bra{\psi}A^2\ket{\psi}-\bra{\psi}A\ket{\psi}^2$, then $\Delta A^2+\Delta B^2\geq 1$.
\end{lemma}

\begin{proof}[Proof of Lemma~\ref{lem:uncertainty}]
  Let $\rho=\sum_\psi p_\psi\ket{\psi}$ be a decomposition of $\rho$ into an ensemble of pure states $\ket{\psi}$. By Lemma~\ref{lem:uncertaintypure}, it holds for each $\psi$ that $\bra{\psi}A\ket{\psi}^2+\bra{\psi}B\ket{\psi}^2\leq 1$, so either $|\bra{\psi}A\ket{\psi}|\leq 1/\sqrt{2}$ or $|\bra{\psi}B\ket{\psi}|<1/\sqrt{2}$. Partition the pure states into sets $\cA,\cB$ so that the former inequality holds for $\psi\in\cA$ and the latter for $\psi\in\cB$. Then either $p_{\cA}:=\sum_{\psi\in\cA}p_\psi\geq 1/2$ or $1-p_{\cA}=p_{\cB}:=\sum_{\psi\in\cB}p_\psi\geq 1/2$. Assume that $p_{\cA}\geq 1/2$; the proof for $p_{\cB}\geq 1/2$ is identical. Then the desired result follows by applying the triangle inequality:
  \begin{equation*}
    |\Tr(A\rho)| = \left|\sum_{\psi\in\cA}p_\psi\bra{\psi}A\ket{\psi}+\sum_{\psi\in\cB}p_\psi\bra{\psi}A\ket{\psi}\right| \leq p_{\cA}\cdot\frac{1}{\sqrt{2}}+p_{\cB}\cdot 1 \leq \frac{1}{2\sqrt{2}}+\frac12.
  \end{equation*}
\end{proof}

We follow prior works such as \cite{eldar_local_2017,anshu_nlts_2023} in estabilishing circuit lower bounds for approximate ground states of $\Ham$ by showing that the measurement distributions of these states are well-spread, in the following sense.

\begin{definition}
  For $\mu,\delta>0$, A probability distribution $D$ over $\bF_2^n$ is \textbf{$(\mu,\delta)$-spread} if there exist $S_0,S_1\subseteq\bF_2^n$ such that $D(S_0)\geq\mu$, $D(S_1)\geq\mu$, and $\dis(S_0,S_1)\geq\delta n$.
\end{definition}

We specifically use following result, which appears as Fact~4 in \cite{anshu_nlts_2023} but is similar to an earlier result of \cite{eldar_local_2017}. Below, for an $n$-qubit state $\psi$, we let $D_X^\psi$ and $D_Z^\psi$ denote the distributions over $\bF_2^n$ obtained by measuring $\psi$ in the $X$ and $Z$ bases respectively.

\begin{lemma}[Circuit lower bound \cite{eldar_local_2017,anshu_nlts_2023}]
  \label{lem:circbound}
  Let $\psi$ be a (possibly mixed) quantum state on $n$ qubits such that the $Z$-measurement distribution $D_Z^\psi$ is $(\mu,\delta)$-spread. Then any circuit (on $\geq n$ qubits) that constructs $\psi$ must have depth at least
  \begin{equation*}
    \frac13\log\left(\frac{\delta^2n}{400\log(1/\mu)}\right).
  \end{equation*}
\end{lemma}

\subsection{Statement of Main Result on NLTS Hamiltonians}
In this section, we state our main technical result, which implies that the code Hamiltonian $\Ham$ for a CSS code with linear distance that exhibits the clustering property is NLTS, that is, its approximate ground states cannot be constructed by constant-depth circuits. Crucially, we only assume that the dimension of the code is positive.

Specifically, our main technical result below shows that the measurement distribution of every approximate ground state of $\Ham$ is well-spread in either the $X$ or $Z$ basis.

\begin{theorem}
  \label{thm:spread}
  Let $\Ham$ be the code Hamiltonian for a $[[n,k,d]]_2$ CSS code $\cC=\CSS(C_X,C_Z)$ of positive dimension $k>0$ that exhibits $(c_1,c_2,\epsilon_0)$-clustering. For any
  \begin{equation}
    \label{eq:epsilon}
    \epsilon < \frac{1}{1000}\cdot\min\left\{\frac{\epsilon_0}{2},\;\frac{c_2}{4c_1},\;\frac{d}{2c_1n}\right\},
  \end{equation}
  let $\rho$ be an $\epsilon$-approximate ground state of $\Ham$, so that $\Tr(\Ham\rho)\leq\epsilon$. Then at least one of $D_X^\rho$ or $D_Z^\rho$ is $(\mu,\delta)$-spread for $\mu=.02$ and $\delta=c_2$.
\end{theorem}

Lemma~\ref{lem:circbound} immediately yields the following corollary.

\begin{corollary}[$\Ham$ is NLTS]
  \label{cor:nlts}
  Define $\epsilon,\mu,\delta$, and $\Ham$ as in Theorem~\ref{thm:spread}. Then no $\epsilon$-approximate ground state of $\Ham$ can be constructed by a circuit of depth less than
  \begin{equation*}
    \frac13\log\left(\frac{\delta^2n}{400\log(1/\mu)}\right)+1.
  \end{equation*}
\end{corollary}

Our proof of Theorem~\ref{thm:spread} is similar to \cite{eldar_local_2017} in that we combine an uncertainty principle with a decoding procedure to obtain uncertainty for approximate ground states. We furthermore use the clustering property of $\cC$ similarly to \cite{anshu_nlts_2023}. As such, the key novel aspect of our proof is the use of a ``decoding'' procedure that handles clusters of approximate ground states which do not correspond to any true codeword.

\subsection{Proof of Well-Spreadness for Approximate Ground States}
In this section, we prove Theorem~\ref{thm:spread}. Throughout this section, we maintain the notation in the statement of Theorem~\ref{thm:spread}, so that $\cC=\CSS(C_X,C_Z)$ is a $[[n,k,d]]_2$ CSS code exhibiting $(c_1,c_2,\epsilon_0)$-clustering, $\rho$ is an $\epsilon$-approximate ground state of $\Ham$ for $\epsilon$ as in~(\ref{eq:epsilon}), and $\mu=.02$, $\delta=c_2$.

\subsubsection{Reducing to Well-Spreadness of Pure States with Small Syndrome}
We will first show that $D_X^\rho$ and $D_Z^\rho$ are mostly supported inside $G_X^{O(\epsilon)}$ and $G_Z^{O(\epsilon)}$ respectively, so that up to a small loss in parameters we may assume they are entirely supported inside these sets. We will also show that it suffices to consider pure states $\psi'=\ket{\psi'}\bra{\psi'}$, rather than arbitrary mixed states $\rho$.

Formally, we may decompose our Hilbert space $\cH=(\bC^2)^{\otimes n}$ into orthogonal subspaces as
  \begin{equation*}
    \cH = \bigoplus_{e_X+C_X\in\bF_2^n/C_X,e_Z+C_Z\in\bF_2^n/C_Z}X^{e_Z}Z^{e_X}\cC,
  \end{equation*}
  where the choices of coset representatives in the above sum does not matter because by definition $X^{c_Z}Z^{c_X}\cC=\cC$ for $c_X\in C_X$, $c_Z\in C_Z$. Observe furthermore that each subspace $X^{e_Z}Z^{e_X}\cC$ is by definition an eigenspace of the code Hamiltonian $\Ham$ with eigenvalue $|H_Xe_X|/2m_X+|H_Ze_Z|/2m_Z$.

  Set
  \begin{equation*}
    \epsilon'=1000\epsilon,
  \end{equation*}
  and let
  \begin{equation*}
    \cC^{\leq\epsilon'} = \bigoplus_{e_X+C_X:|H_Xe_X|\leq\epsilon'm_X,e_Z+C_Z:|H_Ze_Z|\leq\epsilon'm_Z}X^{e_Z}Z^{e_X}\cC.
  \end{equation*}
  Therefore $\cC^{\leq\epsilon'}$ is the span of some of the eigenspaces of energy $\leq\epsilon'$, and contains all of the eigenspaces of energy $\leq\epsilon'/2$. Let $\Pi_{\cC^{\leq\epsilon'}}$ denote projection onto this subspace. Note that by definition, every $\ket{\psi'}\in\cC^{\leq\epsilon'}$ has $\supp(D_X^{\psi'})\subseteq G_X^{\epsilon'}$ and $\supp(D_Z^{\psi'})\subseteq G_Z^{\epsilon'}$.

We now reduce the task of proving Theorem~\ref{thm:spread} to the following proposition. Below, recall that we carry the definitions of $\Ham,\rho,\epsilon,\mu,\delta$ from Theorem~\ref{thm:spread}.

\begin{proposition}
  \label{prop:purespread}
  There exist sets $S_X^0,S_X^1,S_Z^0,S_Z^1\subseteq\bF_2^n$ such that $\dis(S_X^0,S_X^1),\dis(S_Z^0,S_Z^1)\geq\delta n$, and such that for every pure state $\ket{\psi'}\in\cC^{\leq\epsilon'}$, either
  \begin{equation}
    \label{eq:purespread}
    D_X^{\psi'}(S_X^0),D_X^{\psi'}(S_X^1)\geq\mu' \hspace{1em}\text{ or }\hspace{1em} D_Z^{\psi'}(S_Z^0),D_Z^{\psi'}(S_Z^1)\geq\mu',
  \end{equation}
  where $\mu'=1/4-1/4\sqrt{2}$.
\end{proposition}

Below, we first prove Theorem~\ref{thm:spread} assuming Proposition~\ref{prop:purespread}; this proof uses relatively standard techniques, though it is slightly tedious. We will subsequently prove the proposition, which contains the key ideas for our result.

\begin{proof}[Proof of Theorem~\ref{thm:spread}]
  Fix any mixed state $\rho$ with $\Tr(\Ham\rho)\leq\epsilon$. Our goal is to use Proposition~\ref{prop:purespread} to show that either
  \begin{equation*}
    D_X^\rho(S_X^0),D_X^\rho(S_X^1)\geq\mu \hspace{1em}\text{ or }\hspace{1em} D_Z^\rho(S_Z^0),D_Z^\rho(S_Z^1)\geq\mu.
  \end{equation*}
  For this purpose, we first decompose $\rho=\sum_\psi p_\psi\ket{\psi}\bra{\psi}$ into a classical ensemble of pure states $\psi$. Then
  \begin{align*}
    \epsilon
    &\geq \Tr(\Ham\rho) \\
    &= \sum_\psi p_\psi\bra{\psi}\Ham\ket{\psi} \\
    &\geq \sum_\psi p_\psi\bra{\psi}\frac{\epsilon'}{2}(I-\Pi_{\cC^{\leq\epsilon'}})\ket{\psi},
  \end{align*}
  where we have used the fact that $\cC^{\leq\epsilon'}$ contains all eigenspaces of $\Ham$ of eigenvalue $\leq\epsilon'/2$. Therefore
  \begin{align}
    \label{eq:psiclosepsip}
    \sum_\psi p_\psi\bra{\psi}\Pi_{\cC^{\leq\epsilon'}}\ket{\psi}
    &= 1-\sum_\psi p_\psi\bra{\psi}(I-\Pi_{\cC^{\leq\epsilon'}})\ket{\psi} \geq 1-\frac{2\epsilon}{\epsilon'}.
  \end{align}
  Proposition~\ref{prop:purespread} implies that~(\ref{eq:purespread}) holds for every $\ket{\psi'}=\Pi_{\cC^{\leq\epsilon'}}\ket{\psi}/\|\Pi_{\cC^{\leq\epsilon'}}\ket{\psi}\|$. Therefore if we let $\Psi_X$ denote the set of $\psi$ for which the first inequality in~(\ref{eq:purespread}) holds for $\ket{\psi'}$ and $\Psi_Z$ the set of $\psi$ for which the second inequality in~(\ref{eq:purespread}) holds for $\ket{\psi'}$, then either
  \begin{equation*}
    \sum_{\psi\in\Psi_X}p_\psi \geq \frac12 \hspace{1em}\text{ or }\hspace{1em} \sum_{\psi\in\Psi_Z}p_\psi \geq \frac12.
  \end{equation*}
  Assume the latter of the two inequalities above holds; the proof for the former is analogous. For $b=0,1$, let $\Pi_Z^b$ denote orthogonal projection onto $\spn\{\ket{y}:y\in S_Z^b\}$. Then by definition
  \begin{align*}
    D_Z^\rho(S_Z^b)
    &= \Tr(\Pi_Z^b\rho) \\
    &= \sum_\psi p_\psi\|\Pi_Z^b\ket{\psi}\|^2 \\
    &\geq \sum_\psi p_\psi(\|\Pi_Z^b\ket{\psi'}\|-\|\Pi_Z^b(\ket{\psi}-\ket{\psi'})\|)^2 \\
    &\geq \sum_\psi p_\psi(\|\Pi_Z^b\ket{\psi'}\|-\|\ket{\psi}-\ket{\psi'}\|)^2 \\
    &\geq \sum_\psi p_\psi\|\Pi_Z^b\ket{\psi'}\|^2 - 2\sum_\psi p_\psi\|\ket{\psi}-\ket{\psi'}\|^2 \\
  \end{align*}
  Now by defininition we can bound the first term on the RHS above by restricting to the sum over $\psi\in\Psi_Z$ and applying~(\ref{eq:purespread}) with $D_Z^{\psi'}(S_Z^b)=\|\Pi_Z^b\ket{\psi'}\|^2$ to obtain
  \begin{align*}
    \sum_\psi p_\psi\|\Pi_Z^b\ket{\psi'}\|^2
    &\geq \sum_{\psi\in\Psi_Z}p_\psi\|\Pi_Z^b\ket{\psi'}\|^2 \geq \frac12\cdot\mu'.
  \end{align*}
  Meanwhile, we can bound the second term by expanding $\|\ket{\psi}-\ket{\psi'}\|^2$ as in inner product to obtain
  \begin{align*}
    2\sum_\psi p_\psi\|\ket{\psi}-\ket{\psi'}\|^2
    &= 4\sum_\psi p_\psi(1-\bra{\psi'}\ket{\psi}) \\
    &= 4\sum_\psi p_\psi\left(1-\sqrt{\bra{\psi}\Pi_{\cC^{\leq\epsilon'}}\ket{\psi}}\right) \\
    &\leq 4\sum_\psi p_\psi\left(\frac{1-\bra{\psi}\Pi_{\cC^{\leq\epsilon'}}\ket{\psi}}{2}\right) \\
    &\leq 2\cdot\frac{2\epsilon}{\epsilon'} \\
    &= \frac{4\epsilon}{\epsilon'},
  \end{align*}
  where the first two equalities above apply the definition of $\ket{\psi'}=\Pi_{\cC^{\leq\epsilon'}}\ket{\psi}/\|\Pi_{\cC^{\leq\epsilon'}}\ket{\psi}\|$, and the second inequality holds by~(\ref{eq:psiclosepsip}). Thus we have shown that
  \begin{align*}
    D_Z^\rho(S_Z^b)
    &\geq \frac{\mu'}{2}-\frac{4\epsilon}{\epsilon'} = \frac18-\frac{1}{8\sqrt{2}}-\frac{4}{1000} > \mu,
  \end{align*}
  as desired.
\end{proof}

\subsubsection{Decoding Clusters of Small-Syndrome States}
To prove Proposition~\ref{prop:purespread}, we begin by using the clustering property of $\cC$ to partition $G_X^{\epsilon'}$ and $G_Z^{\epsilon'}$ into clusters, for which we will subsequently choose representative elements that we will use to define a decoding map for states $\ket{\psi'}$ with small syndrome. Below, we first analyze the clustering of $G_Z^{\epsilon'}$; the case of $G_X^{\epsilon'}$ will be exactly analogous.

We consider clusters defined similarly as in \cite{anshu_nlts_2023}. However, to obtain our improvement over \cite{anshu_nlts_2023}, we will leverage an additional linear structure in the set of clusters (Property~2 in Lemma~\ref{it:clusttrans} below), which ultimately allows us to use the uncertainty principle in Lemma~\ref{lem:uncertainty}.

Given $y\in G_Z^{\epsilon'}$, define a cluster $\clustZ{y}\subseteq G_Z^{\epsilon'}$ by
\begin{equation*}
  \clustZ{y} = \{y'\in G_Z^{\epsilon'}:|y+y'|_{C_X^\perp}\leq 2c_1\epsilon'n\}.
\end{equation*}
The following lemma follows directly from our definitions.

\begin{lemma}
  \label{lem:clust}
  The clusters $\clustZ{y}$ for $y\in G_Z^{\epsilon'}$ form a partition of $G_Z^{\epsilon'}$ satisfying the following properties:
  \begin{enumerate}
  \item\label{it:clustfar} Every pair of distinct clusters $\clustZ{y}\neq \clustZ{y'}$ satisfies $\dis(\clustZ{y},\clustZ{y'})\geq c_2n$.
  \item\label{it:clusttrans} For $c\in C_Z$, then $\clustZ{y+c}=\clustZ{y}+c$, and in particular $\clustZ{y+c}=\clustZ{y}$ if and only if $c\in C_X^\perp$.
  \end{enumerate}
\end{lemma}
\begin{proof}
  We first show that the cluters form a parition of $G_Z^{\epsilon'}$. Fix some $y\in G_Z^{\epsilon'}$. Then for every $y'\in \clustZ{y}$, it follows that every $y''\in \clustZ{y'}$ has $|y''+y|_{C_X^\perp}\leq|y''+y'|_{C_X^\perp}+|y'+y|_{C_X^\perp}\leq 4c_1\epsilon'n$. But by assumption (see the statement of Theorem~\ref{thm:spread}) $2\epsilon'<\epsilon_0$ and $4c_1\epsilon'<c_2$, so because $y''+y\in G_Z^{2\epsilon'}$, the clustering property implies that $|y''+y|_{C_X^\perp}\leq 2c_1\epsilon'n$, so that $y''\in \clustZ{y}$. Thus we have shown that every $y'\in \clustZ{y}$ has $\clustZ{y'}\subseteq \clustZ{y}$, and by the same reasoning $\clustZ{y}\subseteq \clustZ{y'}$, so $\clustZ{y'}=\clustZ{y}$. Thus every pair of clusters is either equal or disjoint, so the clusters $\clustZ{y}$ form a partition of $G_Z^{\epsilon'}$.

   Now every pair of distinct clusters $\clustZ{y}\neq \clustZ{y'}$ satisfies $\dis(\clustZ{y},\clustZ{y'})\geq c_2n$, as if this distance was $<c_2n$, the clustering property would imply that it is $\leq 2c_1\epsilon'n$, which then implies that $\clustZ{y}=\clustZ{y'}$.

   It remains to show Property~\ref{it:clusttrans} in the lemma statement. For every $y\in G_Z^{\epsilon'}$ and $c\in C_Z$, by definition $H_Z(y+c)=H_Zy$ and thus $\clustZ{y+c}$ is also a cluster in $G_Z^{\epsilon'}$, which is isomorphic to $\clustZ{y}$ under the isomorphism $y'\mapsto y'+c$; that is, $\clustZ{y+c}=\clustZ{y}+c$. If $c\in C_X^\perp$, then $|y+(y+c)|_{C_X^\perp}=0$ so that $y+c\in\clustZ{y}$ and therefore $\clustZ{y+c}=\clustZ{y}$. Meanwhile, if $c\in C_Z\setminus C_X^\perp$, then $\clustZ{y+c}\neq\clustZ{y}$, as otherwise it would follow that $|y+(y+c)|_{C_X^\perp}=|c|_{C_X^\perp}\leq 2c_1\epsilon' n$. But by assumption (see Theorem~\ref{thm:spread}) $\cC$ has distance $d>2c_1\epsilon'n$, so $|c|_{C_X^\perp}>2c_1\epsilon'n$.
\end{proof}

Lemma~\ref{lem:clust} implies that $\clustZ{y}$ has distinct translates by all $c+C_X^\perp\in C_Z/C_X^\perp$, where all representatives of a given coset of $C_X^\perp$ yield the same translate. We denote the collection of these translates for a given cluster $\clustZ{y}$ by
\begin{equation*}
  \transZ{\clustZ{y}} = \{\clustZ{y+c}:c\in C_Z\}.
\end{equation*}
For each such collection $\transZ{}=\transZ{\clustZ{y}}$ of clusters, we fix an arbitrary representative $\clustZ{}(\transZ{})\in\transZ{}$.

Now for a given syndrome $s=H_Zy\in\bF_2^{m_Z}$ of some $y\in G_Z^{\epsilon'}$, so that $|s|\leq\epsilon'm_Z$, then the set of bit strings with syndrome $s$ is precisely the coset $y+C_Z$. By Lemma~\ref{lem:clust}, $(y+C_Z)\cap\clustZ{}(\transZ{\clustZ{y}})$ is a coset of $C_X^\perp$, and is in particular therefore nonempty. Thus we may associate to $s$ an arbitrary representative $e_Z(s)\in(y+C_Z)\cap\clustZ{}(\transZ{\clustZ{y}})$.

We now let $\Dec_Z$ be a unitary acting on $n+m_Z$ qubits with the following ``decoding'' property: for every $y\in G_Z^{\epsilon'}$, it holds that
\begin{equation*}
  \Dec_Z\ket{y}\otimes\ket{0}=\ket{y+e_Z(H_Zy)}\otimes\ket{H_Zy}.
\end{equation*}
We let $\Dec_Z^1$ be the channel acting on $n$ qubits that simply applies $\Dec_Z$ and traces out the syndrome register. Formally, for every $y\in G_Z^{\epsilon'}$, then
\begin{equation*}
  \Dec_Z^1(\ket{y}\bra{y})=\ket{y+e_Z(H_Zy)}\bra{y+e_Z(H_Zy)}.
\end{equation*}
Equivalently, $\Dec_Z^1$ is the channel that performs a $Z$-syndrome measurement on its input $\ket{y}$, and then adds $e_Z(s)$ to the post-measurement state, where $s=H_Zy$ was the measurement outcome.

The channel $\Dec_Z^1$ performs a weak form of decoding in the following sense. Recall that an ordinary decoder in the $Z$ basis for a CSS code maps all bit strings near a given codeword $c\in C_Z$ to some element of the coset $c+C_X^\perp$. In contrast, as shown below, the key property of our decoding channel $\Dec_Z^1$ is that it sends all bit strings in a given cluster $\clustZ{y}$ to elements of the same coset $c+C_X^\perp\in C_Z/C_X^\perp$, though this coset may be far away from the cluster $\clustZ{y}$.

\begin{lemma}
  \label{lem:decclust}
  For every cluster $\clustZ{y}$ and every pair of elements $y,y'\in\clustZ{y}$, then $y'+e_Z(H_Zy')\in y+e_Z(H_Zy)+C_X^\perp$.
\end{lemma}
\begin{proof}
  By definition $e_Z(H_Zy)\in y+C_Z$ and $e_Z(H_Zy')\in y'+C_Z$ both lie in the cluster $\clustZ{}(\transZ{\clustZ{y}})$. Therefore by Lemma~\ref{lem:clust}, both $e_Z(H_Zy')$ and $y'+(y+e_Z(H_Zy))$ belong to both $y'+C_Z$ and $\clustZ{}(\transZ{\clustZ{y}})$, and thus $e_Z(H_Zy')\in y'+(y+e_Z(H_Zy))+C_X^\perp$, or equivalently, $y'+e_Z(H_Zy')\in y+e_Z(H_Zy)+C_X^\perp$.
\end{proof}

To conclude this section, we extend all of the clustering terminology and results above for $G_Z^{\epsilon'}$ to their analogues for $G_X^{\epsilon'}$. Specifically, we similarly obtain a partition of $G_X^{\epsilon'}$ into clusters $\clustX{y}$ for $y\in G_X^{\epsilon'}$. We again conclude that each cluster $\clustX{y}$ in $G_X^{\epsilon'}$ has a set $\transX{\clustX{y}}$ of distinct translates by all $c+C_Z^\perp\in C_X/C_Z^\perp$. We fix arbitrary representative clusters $\clustX{}(\transX{})\in\transX{}$, and assign to each syndrome $s=H_Xy$ for $y\in G_X^{\epsilon'}$ an element $e_X(s)\in(y+C_X)\cap\clustX{}(\transX{\clustX{y}})$. We then obtain an $X$ decoding unitary $\Dec_X$ and channel $\Dec_X^1$, which are defined analogously to their $Z$ analogues, except the syndrome measurement and error correction steps are performed in the $X$ basis instead of the $Z$ basis. Observe that $\Dec_X^1$ and $\Dec_Z^1$ commute, so we can define $\Dec^1=\Dec_X^1\Dec_Z^1=\Dec_Z^1\Dec_X^1$.

\subsubsection{Applying Decoding to Prove Well-Spreadness}
We now complete the proof of Proposition~\ref{prop:purespread}, which as shown above in turn implies Theorem~\ref{thm:spread}, by applying the uncertainty principle in Lemma~\ref{lem:uncertainty} to the decodings of small-syndrome states for $\Ham$.

\begin{proof}[Proof of Proposition~\ref{prop:purespread}]
  Because $\cC$ has dimension $k>0$, the space $C_Z/C_X^\perp=(C_X/C_Z^\perp)^\perp$ is nonzero, so there exist $\bar{c}_X\in C_X\setminus C_Z^\perp$, $\bar{c}_Z\in C_Z\setminus C_X^\perp$ such that $\bar{c}_X\cdot \bar{c}_Z=1$. Fix an arbitrary pair of such elements $\bar{c}_X,\bar{c}_Z$, so that $\bar{X}:=X^{\bar{c}_Z}$ and $\bar{Z}:=Z^{\bar{c}_X}$ are anticommuting logical operators for the code $\cC$.

  For $b=0,1$, define
  \begin{align*}
    S_X^b &= \{y\in G_X^{\epsilon'}:\bar{c}_Z\cdot(y+e_X(H_Xy))=b\} \\
    S_Z^b &= \{y\in G_Z^{\epsilon'}:\bar{c}_X\cdot(y+e_Z(H_Zy))=b\}.
  \end{align*}
  Then by Lemma~\ref{lem:decclust}, for a given cluster $\clustZ{y}$ in $G_Z^{\epsilon'}$, all $y'\in\clustZ{y}$ have $y'+e_Z(H_Zy')$ lying in the same coset $y+e_Z(H_Zy)+C_X^\perp$, and thus all $y'\in\clustZ{y}$ have the same value of $\bar{c}_X\cdot(y'+e_Z(H_Zy'))=\bar{c}_X\cdot(y+e_Z(H_Zy))$. Therefore all $y'\in\clustZ{y}$ lie in the same set $S_Z^b$, where $b=\bar{c}_X\cdot(y+e_Z(H_Zy))$. It follows from Lemma~\ref{lem:clust} that $\dis(S_Z^0,S_Z^1)\geq c_2n=\delta n$. Analogous reasoning implies that $\dis(S_X^0,S_X^1)\geq c_2n=\delta n$.

  It remains to be shown that~(\ref{eq:purespread}) holds for every $\ket{\psi'}\in\cC^{\epsilon'}$. By Lemma~\ref{lem:uncertainty}, either $|\Tr(\bar{X}\Dec^1(\psi'))|\leq 1/2+1/2\sqrt{2}$ or $|\Tr(\bar{Z}\Dec^1(\psi'))|\leq 1/2+1/2\sqrt{2}$. Assume the latter; the proof for the former is analogous. Now because $\bar{Z}$ by definition commutes with $\Dec_X$, the distribution from measuring $\bar{Z}$ on $\Dec^1(\psi')=\Dec_X^1\Dec_Z^1(\psi')$, or equivalently on $\Dec_X(\Dec_Z^1(\psi'\otimes\ket{0}\bra{0}))\Dec_X^\dagger$, is the same as the distribution from measuring $\bar{Z}$ on $\Dec_Z^1(\psi')$. Therefore $|\Tr(\bar{Z}\Dec_Z^1(\psi'))|\leq 1/2+1/2\sqrt{2}$. But by definition if we expand $\ket{\psi'}=\sum_{y\in G_Z^{\epsilon'}}\psi'_y\ket{y}$, then it follows that
  \begin{align*}
    \frac12+\frac{1}{2\sqrt{2}}
    &\geq \Tr(\bar{Z}\Dec_Z^1(\psi')) \\
    &= \Tr((\bar{Z}\otimes I)\Dec_Z(\ket{\psi'}\bra{\psi'}\otimes\ket{0}\bra{0})\Dec_Z^\dagger) \\
    &= \left(\bra{\psi'}\otimes\bra{0}\Dec_Z^\dagger\right)\left((\bar{Z}\otimes I)\Dec_Z\ket{\psi'}\otimes\ket{0}\right) \\
    &= \left(\sum_{y\in G_Z^{\epsilon'}}\bra{y+e_Z(H_Zy)}\otimes\bra{H_Zy}(\psi_y')^\dagger\right) \\
    &\hspace{2em} \cdot\left(\sum_{y\in G_Z^{\epsilon'}}\psi_y'(-1)^{\bar{c}_X\cdot(y+e_Z(H_Zy))}\ket{y+e_Z(H_Zy)}\otimes\ket{H_Zy}\right) \\
    &= \sum_{y\in G_Z^{\epsilon'}}(-1)^{\bar{c}_X\cdot(y+e_Z(H_Zy))}|\psi'_y|^2 \\
    &= \left|\sum_{y\in S_Z^0}|\psi'_y|^2-\sum_{y\in S_Z^1}|\psi'_y|^2\right| \\
    &= |D_Z^{\psi'}(S_Z^0)-D_Z^{\psi'}(S_Z^1)|.
  \end{align*}
  Then because $D_Z^{\psi'}$ is supported inside $G_Z^{\epsilon'}=S_Z^0\sqcup S_Z^1$ by the definition of $\psi'$, it follows that $D_Z^{\psi'}(S_Z^0)+D_Z^{\psi'}(S_Z^1)=1$, so we must have
  \begin{align*}
    D_Z^{\psi'}(S_Z^0),D_Z^{\psi'}(S_Z^1)
    &\geq \frac14-\frac{1}{4\sqrt{2}} = \mu',
  \end{align*}
  as desired.
\end{proof}

\section{Planting Codewords in QLDPC Codes}
\label{sec:plant}
In this section, we show how to plant a nontrivial codeword in the quantum Tanner codes of \cite{leverrier_quantum_2022-1}, thereby ensuring the code has positive dimension regardless of other parameters in the instantiation. For instance, when the inner codes $C_A,C_B$ are chosen to be of rate $1/2$ in the quantum Tanner code construction, the only prior method for bounding dimension, namely by counting parity checks, fails to ensure the dimension of the global code is positive (see Section~\ref{sec:qtan}). However, our planted construction of quantum Tanner codes has positive dimension regardless of the rates of the inner codes, and thus provides a new way to ensure positive dimension, that works in previously unfeasible parameter regimes. We remark that a similar technique also works for the codes of \cite{panteleev_asymptotically_2022}, though we do not present the details to avoid redundancy.

Using the strongly explicit nature of the planted codeword, we apply our construction to improve upon the explicit SoS lower bounds of \cite{hopkins_explicit_2022-1} to obtain \textit{strongly} explicit SoS lower bounds.

\subsection{Intuition: Planted Classical Tanner Codes}
\label{sec:cplant}
In this section, we present the simpler case of how to plant a codeword in a classical Tanner code, which motivates our construction in the quantum case.

Recall that a classical Tanner code $C=\Tan(\Gamma,C_{\text{in}})\subseteq\bF_q^n$ is constructed from a $\Delta$-regular graph $\Gamma$ with $n$ edges and an inner code $C_{\text{in}}\subseteq\bF_q^\Delta$ as follows. We associate the set of all edges in $\Gamma$ with the set $[n]$, and we associate the set of edges incident to each vertex $v\in\Gamma$ with the set $[\Delta]$. Then we define $C$ to be the set of all edge labelings $y\in\bF_q^n=\bF_q^{E(\Gamma)}$ such that the labels of edges incident to each $v\in\Gamma$ form a codeword in $C_{\text{in}}$.

The standard method for ensuring that the rate $R$ of $C$ is positive (and in fact linear in $n$) is to require that $C_{\text{in}}$ be a linear code of rate $R_{\text{in}}>1/2$, so that by counting linear constraints it follows that $R\geq 1-2(1-R_{\text{in}})$.

However, if we only care about ensuring that $R>0$, we may instead simply require that $C_{\text{in}}$ contains the all-$1$s vector $\allones\in\bF_q^\Delta$, as then by definition the global all-$1$s vector $\allones\in\bF_q^n$ must lie in $C$. If the resulting ``planted'' classical code has no other nontrivial codewords, it is simply a repetition code, which is typically uninteresting classically.

However, below we construct a quantum analogue of such planted codes, which are more difficult to construct than their classical counterparts, and yield interesting complexity theoretic applications regardless of their rate. For instance, because the planted codeword is trivial to describe and therefore strongly explicit, we improve the explicit SoS lower bounds of \cite{hopkins_explicit_2022-1} to be strongly explicit. Furthermore, in Corollary~\ref{cor:nlts} of Section~\ref{sec:nlts}, we showed that such qLDPC codes of arbitrarily small rate also yield NLTS Hamiltonians.


\subsection{Construction of Planted Quantum Tanner Codes}
\label{sec:qplant}
In this section, we present our construction of planted quantum Tanner codes. This construction can be viewed as a quantum analogue of the planted classical Tanner codes described in Section~\ref{sec:cplant}. The quantum case requires significantly more care, as desribed below.

The following proposition presents our paradigm for planting a nontrivial codeword in a quantum Tanner code

\begin{proposition}
  \label{prop:qTanplanted}
  Let $\cC$ be a quantum Tanner code as defined in Section~\ref{sec:qTanconstruct} such that the following hold:
  \begin{enumerate}
  \item\label{it:planted} The all-$1$s vector $\allones\in\bF_q^\Delta$ lies in $C_A$ and in $C_B^\perp$.
  \item\label{it:relprime} $n=|Q|=|G||A||B|$ is relatively prime with $q$.
  \end{enumerate}
  Then the all-$1$s vector $\allones\in\bF_q^Q$ lies in $C_Z\setminus C_X^\perp$ and in $C_X\setminus C_Z^\perp$.
\end{proposition}
\begin{proof}
  Because $\allones\in C_A$, the components in every given codeword of $C_A^\perp$ sum to $0$. Therefore every codeword in $C_A^\perp\otimes C_B^\perp$, and thus also in $C_Z^\perp$, has components summing to $0$, as $C_Z^\perp$ is by definition spanned by codewords in $C_A^\perp\otimes C_B^\perp$ supported in neighborhoods of vertices in the square Cayley complex. Thus as the components of $\allones\in\bF_q^Q$ sum to $n\neq 0$ in $\bF_q$ because $n$ is relatively prime with $q$, it follows that $\allones\notin C_Z^\perp$.

  However, as $\allones\in\bF_q^\Delta$ lies in $C_B^\perp$, it follows that $\allones\in\bF_q^{\Delta\times\Delta}$ lies in $(C_A\otimes C_B)^\perp$, and thus $\allones\in\bF_q^Q$ lies in $C_X=\Tan(\Gamma_0,(C_A\otimes C_B)^\perp)$.

  Thus we have shown that $\allones\in C_X\setminus C_Z^\perp$. Analogous reasoning shows that $\allones\in C_Z\setminus C_X^\perp$.
\end{proof}

To instantiate the construction in Proposition~\ref{prop:qTanplanted} such that Condition~\ref{it:planted} is satisfied, we will choose $C_A$ and $C_B^\perp$ at random from the set of codes of some constant rate that contain $\allones$. The following result, which we prove in Section~\ref{sec:prodexproof}, shows that such random ``planted'' codes are still product-expanding, thereby providing a planted analogue of Corollary~\ref{cor:randomprod}.

\begin{proposition}
  \label{prop:randomprodplanted}
  Fix any finite field $\bF_q$. For every fixed $\epsilon>0$, there exists a constant $\rho=\rho(\epsilon)>0$ and a function $\delta(n)=\delta(n;\epsilon)\rightarrow 0$ as $n\rightarrow\infty$ such that the following holds. For every pair of integers $k_1,k_2\in(\epsilon n,(1-\epsilon)n)$, if $C_i\subseteq\bF_q^n$ for $i=1,2$ is drawn uniformly at random from the set of linear codes of dimension $k_i$ that contain $\allones\in\bF_q^n$, then with probability $\geq 1-\delta(n)$ both $(C_1,C_2^\perp)$ and $(C_1^\perp,C_2)$ will be $\rho$-product-expanding.
\end{proposition}

Meanwhile, to ensure that Condition~\ref{it:relprime} in Proposition~\ref{prop:qTanplanted} is satisfied, we will choose the graphs $\Cay(G,A),\Cay(G,B)$ to be almost-Ramanujan graphs from the family given by Theorem~\ref{thm:lwexp}, which we restate below and prove in Section~\ref{sec:expconst}:

\lwexp*

Combining the results above, we immediately obtain the following strongly explicit construction of quantum Tanner codes with a planted all-1s vector.

\begin{theorem}[Planted quantum Tanner codes]
  \label{thm:plantedresult}
  For every finite field $\bF_q$, there exist constants $c_1,c_2>0$ such that there is a strongly explicit infinite family $(\cC^{(n)})_{n\rightarrow\infty}$ of quantum LDPC CSS codes for which every $\cC^{(n)}=\CSS(C_X^{(n)},C_Z^{(n)})$ with $C_X^{(n)},C_Z^{(n)}\subseteq\bF_q^n$ has the following properties:
  \begin{enumerate}
  \item $\cC^{(n)}$ has $(c_1,c_2)$-small-set boundary and coboundary expansion (and therefore has distance $\geq c_1n$ by Lemma~\ref{lem:ssexptodis}).
  \item The all-1s vector $\allones\in\bF_q^n$ lies in $C_X^{(n)}\setminus {C_Z^{(n)}}^\perp$ and in $C_Z^{(n)}\setminus {C_X^{(n)}}^\perp$.
  \end{enumerate}
  In particular, for a sufficiently large constant $\Delta$ and a sufficiently small constant $\rho>0$, such a family $(\cC^{(n)})_{n\rightarrow\infty}$ is given by quantum Tanner codes, where we choose $\Cay(G,A)=\Cay(G,B)$ from a strongly explicit family of $\Delta$-regular almost-Ramanujan graphs given by Theorem~\ref{thm:lwexp}, and the inner codes $C_A,C_B\subseteq\bF_q^\Delta$ are found by a brute force search to ensure that $\allones\in C_A,C_B^\perp$ and that $(C_A,C_B),(C_A^\perp,C_B^\perp)$ are $\rho$-product expanding.
\end{theorem}
\begin{proof}
  By Proposition~\ref{prop:randomprodplanted}, if $\rho>0$ is sufficiently small and $\Delta>0$ is sufficiently large then we can find codes $C_A,C_B\subseteq\bF_q^\Delta$ satisfying the criteria in the theorem statement, namely that $\allones\in C_A,C_B^\perp$ and that $(C_A,C_B),(C_A^\perp,C_B^\perp)$ are $\rho$-product expanding.

  Furthermore, choose any fixed prime $p$ relatively prime with $q$, and let $\Delta_0$ be sufficiently large and $c_1,c_2>0$ be sufficiently small constants such that Theorem~\ref{thm:qTanexp} implies that all $\cC^{(n)}$ have $(c_1,c_2)$-small-set boundary and coboundary expansion as long as $\Cay(G,A)=\Cay(G,B)$ is chosen to be an almost-Ramanujan graph of degree $\Delta\geq\Delta_0$ from the family $(\Gamma_{m,\Delta})_{m\in\bN}$ given by Theorem~\ref{thm:lwexp}, where $|V(\Gamma_{m,\Delta})|=p^{3m}$. Because Theorem~\ref{thm:lwexp} guarantees that $\mathbf{\Delta}\subseteq\bN$ is infinitely large and consists of the union of pairs of consecutive integers, we can always find some sufficiently large $\Delta\in\mathbf{\Delta}$ that is relatively prime with $q$ and that satisfies $\Delta\geq\Delta_0$. Thus we indeed obtain the desired Cayley expanders $\Cay(G,A)=\Cay(G,B)=\Gamma_{m,\Delta}$ for $m\in\bN$, where $|G||A||B|=p^{3m}\Delta^2$ is relatively prime with $q$.

  Now we have shown that the instantiation of quantum Tanner codes above satisfies Conditions~\ref{it:planted}, \ref{it:relprime} in Proposition~\ref{prop:qTanplanted}, so this proposition implies that $\allones\in\bF_q^n$ lies in $C_X^{(n)}\setminus {C_Z^{(n)}}^\perp$ and in $C_Z^{(n)}\setminus {C_X^{(n)}}^\perp$. Meanwhile, as described above, Theorem~\ref{thm:qTanexp} implies that $\cC^{(n)}$ has $(c_1,c_2)$-small-set boundary and coboundary expansion.


  Because the almost-Ramanujan graphs in Theorem~\ref{thm:lwexp} are strongly explicit, and the inner codes $C_A,C_B$ have constant size beacuse $\Delta$ is constant as $n\rightarrow\infty$, the parity check matrices $H_X^{(n)},H_Z^{(n)}$ for $C_X^{(n)},C_Z^{(n)}$ respectively are strongly explicit, which by definition means that $\cC^{(n)}$ is strongly explicit.
\end{proof}


In Theorem~\ref{thm:plantedresult}, we may choose $C_A,C_B$ to have any fixed rates $0<R_A,R_B<1$ for sufficiently large $\Delta$. Because $\cC$ has rate $R\geq-(1-2R_A)(1-2R_B)$ (see Section~\ref{sec:locdim}), it follows that we can in fact ensure that the the codes in Theorem~\ref{thm:plantedresult} have any desired constant rate $0<R<1$.

However, our construction alternatively allows us to obtain quantum Tanner codes of positive dimension for $R_A,R_B$ in previously impossible parameter regimes. Because Theorem~\ref{thm:plantedresult} ensures that $\allones\in C_Z\setminus C_X^\perp$, it follows that $\cC$ always has dimension $\dim(\cC)=\dim(C_Z)-\dim(C_X^\perp)\geq 1$, even when we choose $R_A,R_B$ to be constants for which the bound $R\geq-(1-2R_A)(1-2R_B)$ is meaningless. For instance, we can choose $R_A=R_B$, or take both $R_A,R_B<1/2$, in which case counting parity checks fails to show that the resulting quantum Tanner code $\cC$ has positive dimension. Nevertheless the planted all-1s vector ensures that even in this case $\cC$ must have dimension $\geq 1$.


\subsection{Proof of Product-Expansion for Planted Inner Codes}
\label{sec:prodexproof}
In this section we prove Proposition~\ref{prop:randomprodplanted}. We begin with the following lemma.

\begin{lemma}
  \label{lem:subcodeexp}
  If $(C_1,C_2)$ is $\rho$-product expanding and $C_1'\subseteq C_1$ is a codimension-1 subcode, then $(C_1',C_2)$ is $\rho^2/2$-product expanding.
\end{lemma}
\begin{proof}
  Fix an arbitrary $x\in C_1'\otimes\bF_q^n+\bF_q^n\otimes C_2$. Our goal is to show that there exists a decomposition $x=c+r$ for some $c\in C_1'\otimes\bF_q^n$ and $r\in\bF_q^n\otimes C_2$ satisfying
  \begin{equation}
    \label{eq:wantprodexp}
    |x| \geq \frac{\rho^2n}{2}(|c|_{\text{col}}+|r|_{\text{row}}).
  \end{equation}
  If $|x|\geq\rho^2n^2$, then any decomposition $x=c+r$ suffices, as the right hand side above is always at most $\rho^2n/2\cdot 2n=\rho^2n^2$.
  
  Therefore assume that $|x|<\rho^2n$. By the $\rho$-product expansion of $(C_1,C_2)$, there exists some decomposition $x=c+r$ for $c\in C_1\otimes\bF_q^n$ and $r\in\bF_q^n\otimes C_2$ such that
  \begin{equation}
    \label{eq:assumeprodexp}
    \rho^2n > |x| \geq \rho n(|c|_{\text{col}}+|r|_{\text{row}}).
  \end{equation}
  Meanwhile, by the definition of $x$, there exists some decomposition $x=c'+r'$ for $c'\in C_1'\otimes\bF_q^n$ and $r'\in\bF_q^n\otimes C_2$. Letting $y=c-c'=r'-r$, then $y\in C_1\otimes\bF_q^n$ and $y\in\bF_q^n\otimes C_2$, so $y\in C_1\otimes C_2=C_1'\otimes C_2+\spn\{a\}\otimes C_2$. Therefore we can decompose $y=w+z$ for $w\in C_1'\otimes C_2$ and $z\in\spn\{a\}\otimes C_2$. It follows that $c=c'+y=(c'+w)+z$, where $c'+w\in C_1'\otimes\bF_q^n$ and $z\in\spn\{a\}\otimes C_2$. That is, every column of $c'+w$ lies in $C_1'$, and every nonzero column of $z$ is a scalar multiple of $a\notin C_1'$. Thus the $i$th column of $c$ is not in $C_1'$ if and only if the $i$th column of $z$ is nonzero. But by Lemma~\ref{lem:exptodis}, $C_2$ has distance $\geq\rho n$, and therefore if $z\in\spn\{a\}\otimes C_2$ is nonzero then $z$ has $\geq\rho n$ nonzero columns, which implies that $c$ has $\geq\rho n$ columns that are not in $C_1'$, so in particular $c$ has $\geq\rho n$ nonzero columns. But this assertion that $|c|_{\text{col}}\geq\rho n$ contradicts~(\ref{eq:assumeprodexp}), so our assumption that $z$ is nonzero must have been false. Therefore $z=0$, so $c=c'+w\in C_1'\otimes\bF_q^n$. Thus $x=c+r$ provides our desired decomposition of $x$ satisfying~(\ref{eq:wantprodexp}), as we have shown that $c\in C_1'\otimes\bF_q^n$ and $r\in\bF_q^n\otimes C_2$ in fact satisfy the stronger inequality~(\ref{eq:assumeprodexp}).
\end{proof}

In light of Lemma~\ref{lem:subcodeexp}, we can reduce the problem of proving Propositive~\ref{prop:randomprodplanted} to proving the following result.

\begin{proposition}
  \label{prop:randomprodhalf}
  For every fixed $\epsilon>0$, there exists a constant $\rho=\rho(\epsilon)>0$ and a function $\delta(n)=\delta(n;\epsilon)\rightarrow 0$ as $n\rightarrow\infty$ such that the following holds. For every pair of integers $k_1,k_2\in(\epsilon n,(1-\epsilon)n)$, if $C_1\subseteq\bF_q^n$ is drawn uniformly at random from the set of all linear codes of dimension $k_1$, and $C_2\subseteq\bF_q^n$ is drawn uniformly at random from the set of linear codes of dimension $k_2$ that contain $\allones\in\bF_q^n$, then with probability $\geq 1-\delta(n)$ the pair $(C_1,C_2)$ will be $\rho$-product-expanding.
\end{proposition}

We first show how Proposition~\ref{prop:randomprodhalf} implies Proposition~\ref{prop:randomprodplanted}, and then we will prove Proposition~\ref{prop:randomprodhalf}

\begin{proof}[Proof of Proposition~\ref{prop:randomprodplanted}]
  Let $C_1,C_2\subseteq\bF_q^n$ be random codes as in the statement of Proposition~\ref{prop:randomprodplanted}. Let $p$ be the probability that a random code in $\bF_q^n$ of dimension $k_1'=k_1-1$ contains $\allones\in\bF_q^n$. We draw a $k_1'$-dimensional subcode $C_1'\subseteq C_1$ as follows:
  \begin{itemize}
  \item With probability $1-p$, $C_1'$ is drawn uniformly at random from the codimension-1 subcodes of $C_1$ that do not contain $\allones$.
  \item With probability $p$, $C_1'$ is drawn uniformly at random from the codimension-1 subcodes of $C_1$ that contain $\allones$.
  \end{itemize}
  Then because $C_1$ is a random $k_1$-dimensional code containing $\allones$, by construction $C_1'$ is a uniformly random $k_1'$-dimensional code (correlated with $C_1$), as can be seen by conditioning on the events that $\allones\in C_1'$ and $\allones\notin C_1'$. Thus Proposition~\ref{prop:randomprodhalf} implies that there is a sufficiently small constant $\rho'=\rho'(\epsilon)>0$ such that $({C_1'}^\perp,C_2)$ is $\rho'$-product-expanding with probability $\rightarrow 1$ as $n\rightarrow\infty$. Lemma~\ref{lem:subcodeexp} then implies that $(C_1^\perp,C_2)$ is $\rho:={\rho'}^2/2$-product-expanding with probability $\rightarrow 1$ as $n\rightarrow\infty$.

  By similar reasoning as above, $(C_1,C_2^\perp)$ is also $\rho={\rho'}^2/2$-product-expanding with probability $\rightarrow 1$ as $n\rightarrow\infty$. Then the desired result follows by union bounding over the events that both $(C_1^\perp,C_2)$ and $(C_1,C_2^\perp)$ are $\rho$-product expanding.
\end{proof}

It remains to prove Proposition~\ref{prop:randomprodhalf}. For this purpose, we adapt the proof of Proposition~\ref{prop:randomprod} given by Kalachev and Panteleev \cite{kalachev_two-sided_2023}. Specifically, \cite{kalachev_two-sided_2023} show that a pair of random codes $(C_1,C_2)$ of any fixed rates $<1$ are product-expanding with high probability for sufficiently large block lengths; we want to show that $(C_1,C_2)$ are still product-expanding with high probability when conditioning on the event that $\allones\in C_2$. To avoid redundancy with \cite{kalachev_two-sided_2023}, we will simply describe the necessary modifications to their proof of Proposition~\ref{prop:randomprod} (Theorem 1 in their paper \cite{kalachev_two-sided_2023}) for this case where we condition on $\allones\in C_2$.

To begin, we recall the following definitions from \cite{kalachev_two-sided_2023}. Below we denote the $q$-ary entropy function by $H_q:[0,1]\rightarrow[0,1]$, so that
\begin{equation*}
  H_q(x) = x\log_q(q-1)-x\log_qx-(1-x)\log_q(1-x).
\end{equation*}
We also say that a subspace $V\subseteq\bF_q^n$ is \textit{$\alpha$-sparse} if it is spanned by vectors of Hamming weight $\leq\alpha n$.

\begin{definition}[Property $(\ast)$ \cite{kalachev_two-sided_2023}]
  A code $C\subseteq\bF_q^n$ of dimension $n-r$ has \textbf{property $(\ast)$} if the following holds for $\alpha=H_q^{-1}(r/8n)$: for every $m\in\{1,\dots,r\}$ and every $\alpha$-sparse $m$-dimensional subspace $V\subseteq\bF_q^n$, then $\dim(C\cap V)<m/2$.
\end{definition}

\cite{kalachev_two-sided_2023} prove Proposition~\ref{prop:randomprod} (their Theorem~1) by consider the following two events $E_1,E_2$ for a pair of random codes $C_1,C_2\subseteq\bF_q^n$ of respective dimensions $k_1,k_2\in(\epsilon n,(1-\epsilon)n)$. Recall here that $\epsilon>0$ is any fixed constant, and $\rho=\rho(\epsilon)>0$ is a sufficiently small constant depending only on $\epsilon$.
\begin{enumerate}
\item $E_1$ is the event that there exists $x\in C_1\otimes\bF_q^\perp+\bF_q^\perp\otimes C_2$ of weight $|x|\leq 2\rho n^2$ and of rank $\rank(x)\geq\epsilon$.
\item $E_2$ is the event that $C_1$ or $C_2$ does not have property ($\ast$).
\end{enumerate}
To prove Proposition~\ref{prop:randomprod}, \cite{kalachev_two-sided_2023} show the following:

\begin{lemma}[\cite{kalachev_two-sided_2023}]
  \label{lem:kpproof}
  If $C_1,C_2\subseteq\bF_q^n$ are random codes of respective dimensions $k_1,k_2\in(\epsilon n,(1-\epsilon)n)$, then for all sufficiently large $n$,
  \begin{enumerate}
  \item\label{it:E1} $\Pr[E_1] \leq 5q^{-\epsilon^2n^2/8}$.
  \item\label{it:E2} $\Pr[E_2] \leq 16q^{-\epsilon n/8}$.
  \item\label{it:exp} If $E_1,E_2$ do not occur then $(C_1,C_2)$ is $\rho$-product-expanding.
  \end{enumerate}
\end{lemma}

Lemma~\ref{lem:kpproof} directly implies Proposition~\ref{prop:randomprod}. Similarly, the following lemma directly implies Proposition~\ref{prop:randomprodhalf}:

\begin{lemma}[\cite{kalachev_two-sided_2023}]
  \label{lem:kpproofF}
  If $C_1,C_2\subseteq\bF_q^n$ are random codes of respective dimensions $k_1,k_2\in(\epsilon n,(1-\epsilon)n)$, and $F$ denotes the event that $\allones\in C_2$, then for all sufficiently large $n$,
  \begin{enumerate}
  \item\label{it:E1F} $\Pr[E_1|F] \leq 5q^{-\epsilon^2n^2/8+n}$.
  \item\label{it:E2F} $\Pr[E_2|F] \leq 16q^{-\epsilon n/64}$.
  \item\label{it:expF} If $E_1,E_2$ do not occur then $(C_1,C_2)$ is $\rho$-product-expanding.
  \end{enumerate}
\end{lemma}

Item~\ref{it:exp} in Lemma~\ref{lem:kpproof} and item~\ref{it:expF} in Lemma~\ref{lem:kpproofF} are identical. Furthermore, as $\Pr[F]\geq q^{-n}$, then $\Pr[E_1|F]\leq\Pr[E_1]/\Pr[F]\leq\Pr[E_1]\cdot q^n$, so item~\ref{it:E1} in Lemma~\ref{lem:kpproof} directly implies item~\ref{it:E1F} in Lemma~\ref{lem:kpproofF}.

Therefore to prove Lemma~\ref{lem:kpproofF} and therefore Proposition~\ref{prop:randomprodhalf}, it remains for us to prove item~\ref{it:E2F} of Lemma~\ref{lem:kpproofF}. That is, it suffices to show that conditioned on $F$, each of $C_1$ and $C_2$ has property $(\ast)$ with probability $\geq 1-8q^{-\epsilon n/64}$. As $C_1$ is a random code, Lemma~5 in \cite{kalachev_two-sided_2023} (which they use to show item~\ref{it:E2} in Lemma~\ref{lem:kpproof}) implies that $C_1$ has property $(\ast)$ with probability $\geq 1-8q^{-\epsilon n/8}$. Thus the following lemma completes the proof of Lemma~\ref{lem:kpproofF}, and therefore of Proposition~\ref{prop:randomprodhalf}, and in turn of Proposition~\ref{prop:randomprodplanted}.

\begin{lemma}
  \label{lem:star}
  If $C\subseteq\bF_q^n$ is drawn uniformly at random from the set of $k=(n-r)$-dimensional codes that contain $\allones\in\bF_q^n$, then $C$ has property $(\ast)$ with probability at least
  \begin{equation*}
    1-\frac{4q^{-r/64}}{1-q^{-r/64}}.
  \end{equation*}
\end{lemma}

To prove Lemma~\ref{lem:star}, we will use the following technical result, which appears as Lemma~3 in \cite{kalachev_two-sided_2023}.

\begin{lemma}[\cite{kalachev_two-sided_2023}]
  \label{lem:intbound}
  For every $v$-dimensional subspace $V\subseteq\bF_q^n$, the probability that a uniformly random $u$-dimensional subspace $U\subseteq\bF_q^n$ has $\dim(U\cap V)\geq w$ is $\leq 4q^{-w(n+w-v-u)}$.
\end{lemma}

We will also need the following result on random planted codes.

\begin{lemma}
  \label{lem:plantedint}
  Let $C\subseteq\bF_q^n$ be drawn uniformly at random from the set of $k$-dimensional codes that contain $\allones\in\bF_q^n$. Then for every fixed $(n-1)$-dimensional subspace $W\subseteq\bF_q^n$ such that $\allones\notin W$, the intersection $C\cap W$ is a uniformly random $(k-1)$-dimensional subspace of $W$.
\end{lemma}
\begin{proof}
  Because $\allones\notin W$ and $\allones\in C$, we must have $\dim(C\cap W)=k-1$, so $C=\spn\{\allones\}\oplus(C\cap W)$. That is, $C$ is uniquely determined from $C\cap W$, and each $(k-1)$-dimensional subspace $W'\subseteq W$ gives a distinct $C=\spn\{\allones,W'\}$. Thus because $C$ is drawn uniformly from its allowable set of vector spaces, $\Pr[C\cap W=W']$ is the same for all $W'$.
\end{proof}

We are now ready to prove Lemma~\ref{lem:star}.

\begin{proof}[Proof of Lemma~\ref{lem:star}]
  We simply modify the proof of Lemma~5 in \cite{kalachev_two-sided_2023} to account for the condition that $\allones\in C$. Recall that $\alpha=H_q^{-1}(r/8n)$. Let $C\subseteq\bF_q^n$ be a uniformly random code that contains $\allones\in\bF_q^n$. We want to show that with high probability over the choice of $C$, it holds for every $m\in\{1,\dots,r\}$ and every $\alpha$-sparse $m$-dimensional subspace $V\subseteq\bF_q^n$ that $\dim(C\cap V)<m/2$.

  Fix such an $m$ and $V$. We will first show that
  \begin{equation}
    \label{eq:dimintbound}
    \Pr[\dim(C\cap V)\geq m/2] \leq 4q^{-\frac{9}{64}mr}.
  \end{equation}
  Union bounding over all $m$ and $V$ with this inequality will then give the desired result.

  To prove~(\ref{eq:dimintbound}), we consider the the cases of small $m$ and large $m$ separately. In both cases, we will use Lemma~\ref{lem:plantedint} to relate $C\cap V$ to the intersection of a truly random code with $V$, so that we can apply Lemma~\ref{lem:intbound} to bound $\Pr[\dim(C\cap V)]\geq m/2$.
  \begin{enumerate}
  \item Assume that $m<1/\alpha$. As $V$ has a basis consisting of $m$ $\alpha$-sparse vectors, all elements of $V$ are supported within $\leq\alpha n\cdot m<n$ components, so there is some component $i\in[n]$ on which all $x\in V$ have $x_i=0$. That is, letting $\Pi_i:\bF_q^n\rightarrow\bF_q$ denote projection onto component $i$, then $V\subseteq\Pi_i^{-1}(0)$.

    Now as $\allones\notin\Pi_i^{-1}(0)$, Lemma~\ref{lem:plantedint} implies that $C\cap\Pi_i^{-1}(0)$ is a uniformly random $(k-1)$-dimensional subspace of $\Pi_i^{-1}(0)$.

    Letting $p$ be the probability that a random $k$-dimensional subspace of $\bF_q^n$ lies inside $\Pi_i^{-1}(0)$, we may sample
    \begin{equation*}
      c \sim \begin{cases}
        \Unif(\Pi_i^{-1}(0)\setminus(C\cap\Pi_i^{-1}(0))),&\text{ with probability }p\\
        \Unif(\Pi_i^{-1}(\bF_q\setminus\{0\})),&\text{ with probability }1-p,
      \end{cases}
    \end{equation*}
    so that the vector space
    \begin{equation*}
      U := \spn\{c,(C\cap\Pi_i^{-1}(0))\}
    \end{equation*}
    is a uniformly random $k$-dimensional subspace of $\bF_q^n$. Then because $V\subseteq\Pi_i^{-1}(0)$ so that
    \begin{align*}
      C\cap V
      &= C\cap\Pi_i^{-1}(0)\cap V \subseteq U\cap V,
    \end{align*}
    it follows by Lemma~\ref{lem:intbound} that
    \begin{align*}
      \Pr[\dim(C\cap V)\geq m/2]
      &\leq \Pr[\dim(U\cap V)\geq m/2] \\
      &\leq 4q^{-m/2(n+m/2-m-k)} \\
      &\leq 4q^{-mr/4},
    \end{align*}
    where the final inequality above holds because $k=n-r$ and $m\in\{1,\dots,r\}$. Thus~(\ref{eq:dimintbound}) holds in this case where $m<1/\alpha$.
  \item Assume that $m\geq 1/\alpha$. Again sampling $U$ as in the $m<1/\alpha$ case above, then $U,C$ are $k$-dimensional spaces with $\dim(U\cap C)\geq k-1$. Thus $\dim(C\cap V)\leq\dim(U\cap V)+1$, so by Lemma~\ref{lem:intbound},
    \begin{align*}
      \Pr[\dim(C\cap V)\geq m/2]
      &\leq \Pr[\dim(U\cap V)\geq m/2-1] \\
      &\leq 4q^{-(m/2-1)(n+m/2-1-m-k)} \\
      &\leq 4q^{-(1/2-\alpha)m(n+(1/2-\alpha)m-m-(n-r))} \\
      &\leq 4q^{-\frac38m\left(r-\frac58m\right)} \\
      &\leq 4q^{-\frac{9}{64}mr}.
    \end{align*}
    where the third inequality above holds because $m\geq 1/\alpha$ and $k=n-r$, and the fourth inequality holds because $\alpha=H_q^{-1}(r/8n)\leq H_q^{-1}(1/8)\leq 1/8$, and the fifth inequality holds because $m\in\{1,\dots,r\}$. Thus~(\ref{eq:dimintbound}) holds in this case where $m\geq 1/\alpha$.
  \end{enumerate}

  Now it is well-known that there are $|\{x\in\bF_q^n:|x|\leq\alpha n\}| \leq q^{H_q(\alpha)n}$ distinct $\alpha$-sparse vectors (see for instance Proposition~3.3.3 of~\cite{guruswami_essential_2022}), so there are at most $q^{mH_q(\alpha)}=q^{mr/8n}$ distinct $\alpha$-sparse $m$-dimensional subspace $V\subseteq\bF_q^n$, where we are using the definition of $\alpha=H_q^{-1}(r/8n)$. Union bounding over all such $V$, it follows from~(\ref{eq:dimintbound}) that the probability that there exists an $\alpha$-sparse $m$-dimensional subspace $V$ with $\dim(C\cap V)\geq m/2$ is at most
  \begin{equation*}
    q^{mr/8n}\cdot 4q^{-\frac{9}{64}mr}=4q^{-mr/64}.
  \end{equation*}
  Finally union bounding over $m\in\{1,\dots,r\}$, we conclude that the probability that that there exists some $m\in\{1,\dots,r\}$ with some $\alpha$-sparse $m$-dimensional subspace $V$ with $\dim(C\cap V)\geq m/2$ is at most
  \begin{equation*}
    \sum_{m=1}^r 4q^{-mr/64} \leq \frac{4q^{-r/64}}{1-q^{-r/64}}.
  \end{equation*}
  Thus $C$ has property $(\ast)$ with probability $\geq 1-4q^{-r/64}/(1-q^{-r/64})$, as desired.
\end{proof}

\subsection{Construction of Strongly Explicit Expanders}
\label{sec:expconst}
In this section, we present the proof of Theorem~\ref{thm:lwexp}, which follows from the expander construction of \cite{lubotzky_groups_1993} along with the expansion amplification technique of \cite{jeronimo_almost_2022}. However, there are some details we need to verify to prove Theorem~\ref{thm:lwexp}, specifically that the construction of \cite{lubotzky_groups_1993} is strongly explicit.

We begin by describing the Cayley expanders given in Example~3.4 of \cite{lubotzky_groups_1993}, which are notable because the number of vertices equals a power of any desired prime $p$. In particular, by choosing an odd prime, we obtain Cayley expanders on an odd number of vertices, which is important for our planted quantum Tanner codes over $\bF_2$ due to Condition~\ref{it:relprime} in Proposition~\ref{prop:qTanplanted}. In contrast, the Ramanujan Cayley graphs of \cite{lubotzky_ramanujan_1988} and \cite{morgenstern_existence_1994} have an even number of vertices.

\subsubsection{Base Expander Construction}
\label{sec:lwbaseexp}
This section presents the Cayley expanders given in Example~3.4 of \cite{lubotzky_groups_1993}. For simplicity we restrict attention to expanders constructed over $SL_2$, but \cite{lubotzky_groups_1993} shows that a similar construction over $SL_k$ for any $k\geq 2$ also yields Cayley expanders.

Fix a prime $p$. For every $t\in\bN$, define
\begin{equation*}
  G(t) = \ker(SL_2(\bZ)\rightarrow SL_2(\bZ/t\bZ)),
\end{equation*}
where the homomorphism on the right hand side above simply maps all matrix entries to their value$\pmod{t}$. \cite{lubotzky_groups_1993} shows that there exists a finite generating set $S=S(p)$ for $G(p)$. They then define the Cayley graphs
\begin{align}
  \label{eq:lwbaseexp}
  \Gamma_m^0 &= \Cay(G(p)/G(p^{m+1}),S),
\end{align}
for which they show the following:

\begin{theorem}[\cite{lubotzky_groups_1993}]
  \label{thm:lwbaseexp}
  There exists a constant $\lambda<\Delta$ such that each Cayley graph $\Gamma_m^0$ given by~(\ref{eq:lwbaseexp}) has spectral expansion $\leq\lambda$. Furthermore, $\Gamma_m^0$ has $p^{3m}$ vertices and is $\Delta=|S|$-regular.
\end{theorem}

\subsubsection{Strong Explicitness of Base Expander Construction}
In this section, we show that the expanders in Section~\ref{sec:lwbaseexp} are strongly explicit.

We begin with the following well-known lemma. For completeness we present a proof.\footnote{This proof is fairly standard, though our presentation follows ideas sugggested by \cite{julien_answer_2013,eberhard_answer_2013}.}

\begin{lemma}[Well known]
  \label{lem:SLquotsurj}
  The map $SL_2(\bZ)\rightarrow SL_2(\bZ/t\bZ)$ is surjective.
\end{lemma}
\begin{proof}
  Given any matrix $A\in SL_2(\bZ/t\bZ)$, we can perform row reduction to obtain a diagonal matrix $A'=\begin{pmatrix}a&0\\0&a^{-1}\end{pmatrix}$ for some $a\in\bZ/t\bZ$, which means that $A$ equals $A'$ times a product of elementary matrices (i.e.~matrices with all $1$s on the diagonal, and with a single nonzero off-diagonal entry). But it also holds that
  \begin{equation*}
    A' = \begin{pmatrix}a&0\\0&a^{-1}\end{pmatrix} = \begin{pmatrix}1&a\\0&1\end{pmatrix}\begin{pmatrix}1&0\\-a^{-1}&1\end{pmatrix}\begin{pmatrix}1&a\\0&1\end{pmatrix}\begin{pmatrix}0&-1\\1&0\end{pmatrix}.
  \end{equation*}
  Thus $A\in SL_2(\bZ/t\bZ)$ is a product of elementary matrices along with (possibly) the matrix $\begin{pmatrix}0&-1\\1&0\end{pmatrix}$. But all of these matrices also belong to $SL_2(\bZ)$, so $A$ is the image under the map $SL_2(\bZ)\rightarrow SL_2(\bZ/t\bZ)$ of the product of these matrices in $SL_2(\bZ)$. Thus this map is surjective, as desired.
\end{proof}

We now show the following lemma, which provides a more tractable characterization of the group $G(p)/G(p^{m+1})$. The proof is fairly standard, but we provide it for completeness.

\begin{lemma}
  \label{lem:lwgroup}
  The natural map $G(p)\rightarrow SL_2(\bZ/p^{m+1})$ induces an isomorphism
  \begin{equation}
    \label{eq:lwgroup}
    G(p)/G(p^{m+1}) \cong \ker(SL_2(\bZ/p^{m+1}\bZ)\rightarrow SL_2(\bZ/p\bZ)).
  \end{equation}
\end{lemma}
\begin{proof}
  Consider the sequence of homormorphisms
  \begin{equation*}
    SL_2(\bZ) \xrightarrow{\phi_1} SL_2(\bZ/p^{m+1}\bZ) \xrightarrow{\phi_2} SL_2(\bZ/p\bZ).
  \end{equation*}
  By Lemma~\ref{lem:SLquotsurj}, $\phi_1$ and $\phi_2\circ\phi_1$ are surjective, so $\phi_2$ is also surjective. By definition $G(p^{m+1})=\ker\phi_1$ and $G(p)=\ker(\phi_2\circ\phi_1)$. Now the surjectivity of $\phi_1$ implies that $\phi_1$ induces an isomorphism
  \begin{equation*}
    \tilde{\phi}_1:SL_2(\bZ)/\ker\phi_1 \xrightarrow{\sim} SL_2(\bZ/p^{m+1}\bZ).
  \end{equation*}
  As $G(p)/G(p^{m+1})=\ker(\phi_2\circ\phi_1)/\ker(\phi_1)$ is a subgroup of $SL_2(\bZ)/\ker\phi_1$, we obtain a restricted isomorphism
  \begin{equation*}
    G(p)/G(p^{m+1}) = \ker(\phi_2\circ\phi_1)/\ker(\phi_1) \xrightarrow{\sim} \tilde{\phi}_1(\ker(\phi_2\circ\phi_1)/\ker(\phi_1)).
  \end{equation*}
  But the right hand side above by definition equals $\ker(\phi_2)$, as it holds that $x\in\tilde{\phi}_1(\ker(\phi_2\circ\phi_1)/\ker(\phi_1))$ if and only if $\phi_2(x)$ is the identity. Thus we have shown that $\phi_1$ induces a natural isomorphism
  \begin{equation*}
    G(p)/G(p^{m+1}) \cong \ker(\phi_2),
  \end{equation*}
  as desired.
\end{proof}

Let $G_m$ denote the group in~(\ref{eq:lwgroup}), so that the expanders of \cite{lubotzky_groups_1993} described in Section~\ref{sec:lwbaseexp} are Cayley graphs $\Cay(G_m,S)$ for $m\in\bN$. The following lemma shows that we can enumerate elements of $G_m$ using 3-tuples of elements of $\bZ/p^m\bZ$.

\begin{lemma}
  \label{lem:lwenum}
  There is a bijection
  \begin{equation*}
    \phi:(\bZ/p^m\bZ)^3\rightarrow G_m=\ker(SL_2(\bZ/p^{m+1}\bZ)\rightarrow SL_2(\bZ/p\bZ))
  \end{equation*}
  given by
  \begin{equation*}
    \phi(a,b,c) = I+p\begin{pmatrix}a&b\\c&d\end{pmatrix}
  \end{equation*}
  for $d\in \bZ/p^m\bZ$ given by
  \begin{equation}
    \label{eq:solveford}
    d = (1+pa)^{-1}(pbc-a).
  \end{equation}
\end{lemma}
\begin{proof}
  By definition $1+pa$ is invertible in $\bZ/p^m\bZ$, so $d$ is well defined, and we have
  \begin{align*}
    \phi(a,b,c) = \det\left(I+p\begin{pmatrix}a&b\\c&d\end{pmatrix}\right) &= (1+pa)(1+pd)-p^2bc = 1,
  \end{align*}
  so indeed $\phi$ maps $(a,b,c)$ to an element of $G_m$. We must verify that $\phi$ is injective and surjective.

  To see that $\phi$ is injective, observe that if $(a,b,c)\neq(a',b',c')$, then $(pa,pb,pc)$ and $(pa',pb',pc')$ are distinct tuples in $(\bZ/p^{m+1})^3$, so $\phi(a,b,c)$ and $\phi(a',b',c')$ are distinct matrices in $(\bZ/p^{m+1}\bZ)^{2\times 2}$.

  To see that $\phi$ is surjective, consider that every matrix $M\in G_m=\ker(SL_2(\bZ/p^{m+1}\bZ)\rightarrow SL_2(\bZ/p\bZ))$ is by definition of the form $M=I+p\begin{pmatrix}a&b\\c&d\end{pmatrix}$ for some $a,b,c,d\in\bZ/p^m\bZ$. Now because this matrix has determinant $1$, we have $(1+a)(1+d)-bc=1$, which simplifies to~(\ref{eq:solveford}), so $M=\phi(a,b,c)$. Thus $\phi$ is surjective, as desired.
\end{proof}

As a side note, Lemma~\ref{lem:lwenum} also recovers the fact that $|G_m|=|\bZ/p^m\bZ|^3=p^{3m}$.

We are now ready to show that the expanders described in Section~\ref{sec:lwbaseexp} are strongly explicit.

\begin{proposition}
  \label{prop:lwstrongexplicit}
  For every fixed $p$, the family of Cayley graphs $(\Gamma_m^0=\Cay(G_m,S))_{m\in\bN}$ presented in Section~\ref{sec:lwbaseexp} is strongly explicit.
\end{proposition}
\begin{proof}
  By Lemma~\ref{lem:lwgroup}, we have $G_m\cong\ker(SL_2(\bZ/p^{m+1}\bZ)\rightarrow SL_2(\bZ/p\bZ))$, and the generating set $S$ consists of a fixed finite set of matrices in $\bZ^{n\times n}$, which can be viewed as matrices in $G_m$ by replacing each entry with its value $\pmod{p^{m+1}}$. Note that $S$ depends on $p$, but here we assume $p$ is fixed, and $S$ does not depend on $m$ (except for our interpretation of its entries as integers$\pmod{p^{m+1}}$).

  Now Lemma~\ref{lem:lwenum} implies that the elements of $G_m$ are index by tuples $(a,b,c)\in(\bZ/p^m\bZ)^3$, or equivalently, by tuples $(a,b,c)\in\{0,1,\dots,p^m-1\}^3$. Furthermore, we can go between the tuple representation $(a,b,c)$ and the matrix representation $\phi(a,b,c)$ defined in Lemma~\ref{lem:lwenum} in $\poly(\log p^m)=\poly(m)$ time, as the conversion simply requires computing $d=(1+pa)^{-1}(pbc-a)$ and then multiplying $\begin{pmatrix}a&b\\c&d\end{pmatrix}$ by $p$ and adding $I$; the reverse conversion (from matrix to tuple) is similarly efficient. Thus we can perform group multiplication and inversion in time $\poly(m)$, as these operations are simply given by matrix multiplication and inversion.

  Thus we have shown that the group operations of $G_m$ run in $\poly(\log|G_m|)$ time, and the generating set $S$ can be computed in constant time, so the family of Cayley graphs $\Gamma_m^0=\Cay(G_m,S)$ for $m\in\bN$ is strongly explicit, as desired.
\end{proof}

\subsubsection{Amplifying the Expansion to Almost-Ramanujan}
In this section, we prove Theorem~\ref{thm:lwexp} by applying the expansion amplification technique of \cite{jeronimo_almost_2022}. In particular, \cite{jeronimo_almost_2022} show the following (see for instance their Theorem~1.2).

\begin{theorem}[\cite{jeronimo_almost_2022}]
  \label{thm:expamp}
  Let $(\Gamma_m^0=\Cay(G_m,S_m))_{m\in\bN}$ be a strongly explicit family of $\Delta_0$-regular Cayley graphs with spectral expansion bounded by some constant $\lambda<\Delta_0$. Then there exists an infinite set $\mathbf{\Delta}\subseteq\bN$ for which there is a strongly explicit family $(\Gamma_{m,\Delta}=\Cay(G_m,S_{m,\Delta}))_{m\in\bN,\Delta\in\mathbf{\Delta}}$ of almost-Ramanujan Cayley graphs, where $\Gamma_{m,\Delta}$ has degree $|S_{m,\Delta}|=\Delta$.
\end{theorem}

Theorem~\ref{thm:lwexp} now follows almost immediately from Proposition~\ref{prop:lwstrongexplicit} and Theorem~\ref{thm:expamp}:

\begin{proof}[Proof of Theorem~\ref{thm:lwexp}]
  We apply Theorem~\ref{thm:expamp} to the Cayley graphs $(\Gamma_m^0=\Cay(G_m,S))_{m\in\bN}$ presented in Section~\ref{sec:lwbaseexp}. These graphs have constant degree by definition, have constant spectral expansion $\lambda<\Delta$ by Theorem~\ref{thm:lwbaseexp}, and are strongly explicit by Proposition~\ref{prop:lwstrongexplicit}. Therefore Theorem~\ref{thm:expamp} gives a strongly explicit family $(\Gamma_{m,\Delta}=\Cay(G_m,S_{m,\Delta}))_{m\in\bN,\Delta\in\mathbf{\Delta}}$ of almost-Ramanujan Cayley graphs over the groups $G_m$ given by~(\ref{eq:lwgroup}), which have order $p^{3m}$.

  It only remains to ensure that $\mathbf{\Delta}$ has the property that if $\Delta\in\mathbf{\Delta}$, then either $\Delta-1\in\mathbf{\Delta}$ or $\Delta+1\in\mathbf{\Delta}$. For this purpose, for each graph $\Gamma_{m,\Delta}=\Cay(G_m,S_{m,\Delta})$ in our family such that $\Delta\in\mathbf{\Delta}$ but $\Delta+1\notin\mathbf{\Delta}$, we may also add the graph $\Gamma_{m,\Delta+1}:=\Cay(G_m,S_{m,\Delta}\cup\{\text{id}\})$ obtained by adding the identity element as a Cayley generator. The resulting family of graphs $\Gamma_{m,\Delta}$ ranges over the possible degrees $\Delta\in\mathbf{\Delta'}:=\mathbf{\Delta}\cup\{\Delta+1:\Delta\in\mathbf{\Delta}\}$, which has the desired property that if $\Delta\in\mathbf{\Delta'}$, then either $\Delta-1\in\mathbf{\Delta'}$ or $\Delta+1\in\mathbf{\Delta'}$.

  We must verify that this larger family $(\Gamma_{m,\Delta})_{m\in\bN,\Delta\in\mathbf{\Delta'}}$ is still almost-Ramanujan. Adding the identity element as a Cayley generator increases the graph degree by $1$, and increases the spectral expansion by at most $1$. This latter claim holds because adding the identity to the Cayley generating set has the effect of adding the identity matrix to the adjacency matrix, which increases all eigenvalues by $1$, and thus increases the spectral expansion by at most $1$. Therefore if the original degree-$\Delta$ graphs $\Gamma_{m,\Delta}$ have spectral expansion at most $\lambda(\Delta)=\Delta^{1/2+o(1)}$, then the added degree-$(\Delta+1)$ graphs have spectral expansion at most $\lambda(\Delta+1):=\lambda(\Delta)+1$, which still grows as $(\Delta+1)^{1/2+o(1)}$. Thus our final augmented family $(\Gamma_{m,\Delta})_{m\in\bN,\Delta\in\mathbf{\Delta'}}$ is almost-Ramanujan, as desired. Note that the strong explicitness of this augmented family is immediate from the strong explicitness of the original family $(\Gamma_{m,\Delta})_{m\in\bN,\Delta\in\mathbf{\Delta}}$.
\end{proof}

\begin{remark}
  The expansion amplification of \cite{jeronimo_almost_2022} in Theorem~\ref{thm:expamp}, along with our technique of adding the identity to to the Cayley generating set in the proof above, assume that the Cayley generating sets of our graphs are actually \textit{multisets}. That is, we allow repeated Cayley generators, so our Cayley graphs are actually multigraphs, meaning there can be multiple distinct edges between the same two vertices.

  Fortunately, all of our applications of these graphs apply equally well to multigraphs and simple (non-multi) graphs. Indeed, just as there are no complications in defining a classical Tanner code on a multigraph, there are no complications in defining a quantum Tanner code using multigraphs; the edge and face sets simply become multisets.
\end{remark}

\subsection{Application to Strongly Explicit Sum-of-Squares Lower Bounds}
In this section, we describe how we use our planted quantum Tanner codes to obtain \textit{strongly} explicit lower bounds against a linear number of levels of the SoS hierarchy, thereby improving upon the weakly explicit SoS lower bounds of Hopkins and Lin~\cite{hopkins_explicit_2022-1}.

Hopkins and Lin~\cite{hopkins_explicit_2022-1} show that quantum LDPC codes with small-set boundary and coboundary expansion yield CSPs that are hard for a linear number of levels of the Sum-of-Squares SDP hierarchy. Specifically, the CSPs they use are instances of $\ell$-LIN over $\bF_2$ (or equivalently, $\ell$-XOR) given in Definition~\ref{def:hlcsp} below.

Recall that in general, an instance of $\ell$-LIN consists of a vector of $m$ variables $y=(y_1,\dots,y_m)$ along with a set of $n$ (affine) linear constraints over $\bF_q$, each of which involves $\leq\ell$ variables. Formally, such a system of equations can be expressed in matrix notation as $Ay=\beta$ for some $A\in\bF_q^{n\times m}$ and some $\beta\in\bF_q^n$, where each row of $A$ has $\leq\ell$ nonzero entries.

Below, we let the \textit{locality} of a CSS code $\cC=\CSS(C_X=\ker H_X,C_Z=\ker H_Z)$ refer to the maximum Hamming weight of any row or column of $H_X$ or $H_Z$. The qLDPC codes we consider by definition have locality $\ell=O(1)$ as $n\rightarrow\infty$.

\begin{definition}[$\ell$-LIN instances from qLDPC codes \cite{hopkins_explicit_2022-1}]
  \label{def:hlcsp}
  Let $\cC=\CSS(C_X=\ker H_X,C_Z=\ker H_Z)$ be a CSS code of locality $\ell$. Also fix any $\beta\in C_X\setminus C_Z^\perp$. Then define the associated $\ell$-LIN instance $\cI_{\cC,\beta}$ to have $m=m_Z$ variables $y_1,\dots,y_m\in\bF_q$ and $n$ linear constraints over $\bF_q$ given by the system of equations $H_Z^\top y=\beta$, where $y=(y_1,\dots,y_m)$.
\end{definition}

\cite{hopkins_explicit_2022-1} instantiates this definition with quantum Tanner codes. Although quantum Tanner codes are strongly explicit, meaning that the matrices $H_X,H_Z$ are strongly explicit, any $\ell$-LIN instance $\cI_{\cC,\beta}$ from these codes requires a description of some $\beta\in C_X\setminus C_Z^\perp$. Previously, the only known method for finding such a codword was via Gaussian elimination, which runs in $\poly(n)$ time, and thus only yields a (weakly) explicit construction of $\beta$ and of $\cI_{\cC,\beta}$.

In contrast, our planted quantum Tanner codes in Theorem~\ref{thm:plantedresult} are guaranteed to have the all-1s vector $\allones\in C_X\setminus C_Z^\perp$, which is by definition strongly explicit. As such, we immediately obtain the following.

\begin{lemma}
  \label{lem:secsp}
  If $\cC$ is chosen from a family of planted quantum Tanner codes from Theorem~\ref{thm:plantedresult} and $\beta=\allones$, then $\cI_{\cC,\beta}$ gives a family of strongly explicit $\ell$-LIN instances for a constant $\ell=O(1)$.
\end{lemma}

Formally, \cite{hopkins_explicit_2022-1} obtain their SoS lower bounds by showing the following result, which they applied to quantum Tanner codes. Below, recall that an $\ell$-LIN instance is \textit{$\mu$-satisfiable} if there exists an assignment of the variables satisfying $\geq\mu$-fraction of the linear constraints. We refer to \cite{hopkins_explicit_2022-1} and the references within for background on the SoS SDP hierarchy.

\begin{theorem}[\cite{hopkins_explicit_2022-1}]
  \label{thm:exptosos}
  Let $\cC=\CSS(C_X=\ker H_X,C_Z=\ker H_Z)$ be a quantum LDPC code of locality $\ell$ with $(c_1,c_2)$-small-set boundary and coboundary expansion over a prime-sized alphabet $\bF_p$. Then for every $\beta\in C_X\setminus C_Z^\perp$, the $\ell$-LIN instance $\cI_{\cC,\beta}$ with $m=m_Z$ variables and $n$ constraints satisfies the following:
  \begin{enumerate}
  \item Soundness: $\cI_{\cC,\beta}$ is at most $(1-c_1)$-satisfiable.
  \item Completeness: $\cI_{\cC,\beta}$ cannot be refuted by $c_1c_2m/4\ell$ levels of the SoS hierarchy.
  \end{enumerate}
\end{theorem}

Although Hopkins and Lin \cite{hopkins_explicit_2022-1} only showed Theorem~\ref{thm:exptosos} for the binary alphabet $\bF_2$, their same proof extends to arbitrary fields $\bF_p$ for prime $p$. Specifically, their proof uses small-set (co)boundary expansion to establish a bound on \textit{refutation complexity}, which was then shown to imply an SoS bound for the binary alphabet $\bF_2$ by Schoenebeck \cite{schoenebeck_linear_2008}, and for prime-sized alphabets $\bF_p$ by Tulsiani \cite{tulsiani_csp_2009}.

Thus as described above, \cite{hopkins_explicit_2022-1} obtained (weakly) explicit, but not strongly explicit, lower bounds against $\Omega(n)$ levels of SoS by taking $\cC$ to be a quantum Tanner code in Theorem~\ref{thm:exptosos}. Meanwhile, applying our planted quantum Tanner codes in Theorem~\ref{thm:plantedresult} with Lemma~\ref{lem:secsp}, we immediately obtain the following corollary to Theorem~\ref{thm:exptosos}.

\begin{corollary}[Strongly explicit SoS lower bounds for $\ell$-LIN]
  \label{cor:sosresult}
  The $\ell$-LIN instances $\cI_{\cC,\allones}$ for planted quantum Tanner codes $\cC$ over any fixed prime-sized alphabet $\bF_p$ provide a family of strongly explicit instances with satisfiability $\leq(1-\Omega(1))$, such that no instance can be refuted by $cn$ levels of the SoS hierarchy for a sufficiently small constant $c>0$.
\end{corollary}

\cite{hopkins_explicit_2022-1} also showed a reduction that used their $\ell$-XOR (i.e.~$\ell$-LIN over $\bF_2$) SoS lower bounds to obtain 3-XOR SoS lower bounds, as stated below. Intuitively, the reduction works by introducing dummy variables to reduce the sizes of constraints.

\begin{proposition}[Follows from Claim~6.5 in \cite{hopkins_explicit_2022-1}]
  \label{prop:lto3}
  Let $(\cI_n)_{n\rightarrow\infty}$ be a strongly explicit family of $\ell$-XOR instances such that each $\cI_n$:
  \begin{enumerate}
  \item has $\Theta(n)$ variables and constraints,
  \item has satisfiability $\leq(1-\Omega(1))$,
  \item cannot be refuted by $cn$ levels of the SoS hierarchy for a sufficiently small constant $c>0$.
  \end{enumerate}
  Then there exists a strongly explicit family $(\cI_n')_{n\rightarrow\infty}$ of 3-XOR instances that also satisfies the three properties above.
\end{proposition}

While \cite{hopkins_explicit_2022-1} only showed that the reduction behind Proposition~\ref{prop:lto3} preserves weak explicitness, it by definition also preserves strong explicitness, so Proposition~\ref{prop:lto3} holds. As a corollary of Proposition~\ref{prop:lto3} and Corollary~\ref{cor:sosresult}, we immediately obtain the following.

\begin{corollary}[Strongly explicit SoS lower bounds for 3-XOR]
  \label{cor:3xor}
  There exists a strongly explicit family $(\cI_n)_{n\rightarrow\infty}$ of 3-XOR instances such that each $\cI_n$ has $\Theta(n)$ variables and constraints, has satisfiability $\leq(1-\Omega(1))$, and cannot be refuted by $cn$ levels of the SoS hierarchy for a sufficiently small constant $c>0$.
\end{corollary}

We suspect a similar reduction should work for $\ell$-LIN over arbitrary fields $\bF_p$, but for conciseness we will not pursue this direction.

\section{Acknowledgments}
We thank Max Hopkins for numerous helpful discussions, and for bringing the problem of strongly explicit SoS lower bounds to our attention. We also thank Venkat Guruswami for helping to improve the exposition.

L.~Golowich is supported by a National Science Foundation Graduate Research Fellowship under Grant No.~DGE 2146752, and in part by V.~Guruswami's Simons Investigator award and UC Noyce Initiative Award award. This work was done in part while the authors were visiting the Simons Institute for the Theory of Computing.

\bibliographystyle{alpha}
\bibliography{library}

\end{document}